\documentclass[11pt]{article}
\usepackage[letterpaper]{geometry}
\usepackage{amsmath,amsthm,amsfonts,amssymb}
\usepackage{enumerate,color,xcolor}
\usepackage{graphicx}
\usepackage{subfigure}
\usepackage[square,numbers]{natbib}
\usepackage{ulem}
\usepackage{cancel}

\numberwithin{equation}{section}
\theoremstyle{plain}

\newtheorem{theorem}{Theorem}

\newtheorem{definition}[theorem]{Definition}

\newtheorem{lemma}[theorem]{Lemma}

\newtheorem{proposition}[theorem]{Proposition}

\newtheorem{remark}[theorem]{Remark}

\begin{document}

\begin{center}
  \Large \bf Asymptotic Analysis for Optimal Dividends in a Dual Risk Model\end{center}

\author{}
\begin{center}
{Arash Fahim}\,\footnote{Department of Mathematics, Florida State University, 1017 Academic Way, Tallahassee, FL-32306, 
United States of America;},
  Lingjiong Zhu\,\footnote{Corresponding Author. Department of Mathematics, Florida State University, 1017 Academic Way, Tallahassee, FL-32306, United States of America; Email: zhu@math.fsu.edu; Tel.: 1-850-644-3800; Fax: 1-850-644-4053.
  }
\end{center}

\begin{center}
 \today
\end{center}

\begin{abstract}
The dual risk model is a popular model in finance and insurance, which
is often used to model the wealth process of a venture capital or high tech company.
Optimal dividends have been extensively studied in the literature for a dual risk model.
It is well known that the value function of this optimal control problem does
not yield closed-form solutions except in some special cases. In this paper, 
we study the asymptotics of the optimal dividend problem when the parameters of the model
go to either zero or infinity. Our results provide  insights to the optimal strategies
and the optimal values when the parameters are extreme. 
\end{abstract}

\section{Introduction}

In a classic risk model in the insurance literature, the surplus process increases
continuously in time
with the constant premium rate
and decreases due to the claims that follow a compound Poisson process.
In a dual risk model (see, e.g., \citet{Avanzi}), the opposite happens, that is, the surplus process
decreases continuously in time with the constant rate,
and increases according to a compound Poisson process.
A dual risk model can be used to model wealth of a venture capital,
where the running cost is deterministic whereas the revenues are stochastic, 
see, e.g., \cite{Afonso, Avanzi, AvanziII, BE,CheungI,CheungII,FZ,Ng,NgII, RCE,YS} etc.

In the pioneering work by  \citet{Avanzi}, they studied
the optimal dividend problem in a dual risk model, 
that is, the optimal dividend strategy
for maximizing the expected payments of all the future dividends
to the shareholders until the time of the ruin.
They proved that the optimal strategy is a barrier strategy, that is, 
there exists an optimal barrier below which no dividend is paid out
and when the wealth process jumps above the barrier,
all the surplus is paid out as the dividends to the shareholders immediately
and the surplus drops to the level of the barrier.

There have been many related works of \citet{Avanzi}. 
In \citet{AvanziII}, they studied a dividend barrier strategy for a dual risk model
whereby dividend decisions are made only periodically, but ruin is still allowed to occur in continuous time.
\citet{Ng} studied a dual model with a threshold dividend strategy, with exponential interclaim times.
In another related work,  \citet{Afonso} studied the connections between dual 
and classical risk models and used that to compute various quantities of interest. 
 \citet{CheungII} considered
the dividend moments in a dual risk model. They derived 
integro-differential equations for the moments of the total discounted dividends which can be solved explicitly 
assuming the jump size distribution has a rational Laplace transform.
The expected discounted dividends assuming
the profits follow a Phase Type distribution were studied in \citet{RCE} . 
The Laplace transform of the ruin
time, expected discounted dividends for the Sparre-Andersen dual model
were derived in \citet{YS}. 
In \citet{DKZ}, the finite time distribution of the ruin time was 
studied for the case when the interarrival times are not independent. 
\citet{BKY1} showed that for all spectrally positive L\'evy processes, the optimal strategy is of barrier type and they provided a closed-from solution for the optimal dividend value function in terms of the inverse Laplace transform of scale functions of the L\'evy process. In a separate paper \cite{BKY2}, they extended their analysis to the case where payment of dividend carries a fixed transaction cost by using double-barrier dividend strategies, i.e., when the dual process hits above the upper barrier, as much dividend is paid as to bring the process to the lower barrier. When there is no transaction cost, the two barriers collapse to one. 
An iterative approach toward finding optimal barrier and value function is adopted in \cite{EL} where they showed via a fixed point theorem that the convergence of an initial guess of the barrier to the optimal one is exponentially fast, as well as the convergence of the value of the initial barrier to the value function.
When there is a random delay
for the innovations turned to profits, the dual risk model
becomes time inhomogeneous and the ruin probabilities and the ruin time distributions
are studied in \citet{ZhuII}. 

Recently, \citet{FZ} considered the optimal investment on research
and development to minimize the ruin probability for a dual risk model. Additional
investment on a risky market index and the generalization to a state-dependent dual
risk model was also considered in \cite{FZ}.

In a dual risk model, except for the special case including when the probability density function of the jump size
in the compound Poisson process 
is an exponential function or a sum of exponential functions, in general, there is no closed-form
formulas for the value function 
and the no closed-form formula for the optimal barrier. 

In this paper, we are going to focus on the asymptotics for the optimal dividend problem
in a dual risk model. Even though the general problem does not yield closed-form formulas,
the asymptotics are very explicit and intuitive. They also provide useful insights
to help us understand better the nature of the optimal dividend problem in the dual risk model.
For the optimal dividend problem in the dual risk model, we know
that the optimal strategy is a barrier strategy. 
But in practice, most of the companies pay quarterly dividends
and the investors prefer continuous yield rather than the barrier dividend payments
for dual risk models. 
One of the interesting discoveries of our paper is
that in many asymptotic regimes, 
one can find a nearly optimal strategy
that pays the continuous dividend yield when the surplus of the company is sufficiently large. 
The nearly optimal strategy we will present does not pay any dividend
until the surplus reaches a large value and then it starts to pay out dividend continuously.
In other words, a start-up company should wait till it becomes a mature company
before paying out dividend continuously. 
Many high-tech companies, after the successful IPO (Initial Public Offering), remain growth companies
and do not pay dividends for a very long period of time, until they get more mature, or sometimes
in response to big shareholders and activists' demand. 
After that, they will start paying dividends continuously
with dividend yields being constant, and the dividend yield usually increases modestly and consistently
year over year. Therefore, the nearly optimal strategy we will present
is perhaps more consistent with practice. 
It also suggests
that the most commonly adopted dividend strategies in the corporate world
may not be optimal, but at least, nearly optimal.
We will show that in some asymptotic regimes,
continuous dividend yield strategy can be nearly optimal, even though not exactly optimal. 

We will introduce the model setup in Section~\ref{sec:model}. The paper has all the main results in Section~\ref{sec:main} and all the proofs in Section~\ref{sec:proofs}.
Before we proceed, let us introduce the standard notations that will be used
throughout the rest of the paper.
The standard notion $f\sim g$ is used to denote $\frac{f}{g}$ has limit equal to $1$.
We use the notation $f=O(g)$ to denote $\limsup\frac{|f|}{g}<\infty$.
Finally, we use the notation $f=o(1)$ to denote that $f$ has limit equal to $0$.

\section{Model Setup}\label{sec:model}

In this paper, we are interested in studying the dual risk model, see, e.g., \cite{Avanzi}, 
where the surplus or wealth process satisfies the dynamics:
\begin{equation}\label{dual:risk:model}
dX_{t}=-\rho dt+dJ_{t},\qquad X_{0}=x>0,
\end{equation}
where $\rho>0$ is the running cost of the company and $J_{t}=\sum_{i=1}^{N_{t}}Y_{i}$, 
is the stream of profits, where $Y_{i}$ are i.i.d. $\mathbb{R}^{+}$ valued random variables
with common probability density function $p(y)$, $y>0$ and $N_{t}$
is a Poisson process with intensity $\lambda>0$. $Y_{i}$'s are often known
as the innovation sizes or the random future revenues, and $\lambda$
can be interpreted as the innovation rate.
In a classic risk model in the insurance literature, the surplus process increases
continuously in time
with the constant premium rate $\rho$
and decreases due to the claims that follow a compound Poisson process $J_{t}$,
whereas in the dual risk model \eqref{dual:risk:model}, the opposite happens, that is, the surplus process
decreases continuously in time with the constant rate $\rho$,
and increases according to a compound Poisson process $J_{t}$.

Let $\tau:=\inf\{t>0:X_{t}\leq 0\}$ be the ruin time of the company.
Under the assumption that $\lambda\mathbb{E}[Y_{1}]>\rho$, it is well known
that the infinite-horizon ruin probability has the formulas
$\mathbb{P}_{x}(\tau<\infty)=e^{-\alpha x}$, where $x$ is the initial wealth $X_{0}$
of the company and $\alpha$ is the unique position value that satisfies the equation (see, e.g., \cite{Avanzi}):
\begin{equation}
\rho\alpha+\lambda\int_{0}^{\infty}[e^{-\alpha y}-1]p(y)dy=0.
\end{equation}
Similarly, one can also compute the Laplace transform
of the ruin time, $\mathbb{E}_{x}[e^{-\delta\tau}]=e^{-\beta x}$,
where $\beta$ is the unique positive value that satisfies the equation (see, e.g., \cite{Avanzi}):
\begin{equation}
\beta\rho+\lambda\int_{0}^{\infty}[e^{-\beta y}-1]p(y)dy-\delta=0.
\end{equation}

In the pioneering work by \citet{Avanzi}, they studied
the optimal dividend problem in the dual risk model \eqref{dual:risk:model}.
Let $\delta>0$ be the discount rate used in the discount factor and
$D_{t}$ be the rate of the dividend payment at time $t$ by the company
to the shareholders. 

Given a dividend payment strategy $D\in\mathcal{D}$, where
 the set of admissible dividend payment strategies $\mathcal{D}$ is the collection of all adapted nondecreasing c\`{a}dl\`{a}g processes, $D_t$ is the cumulated dividend until time $t$, and the wealth process is given by
\begin{equation}
dX_{t}=-\rho dt-dD_{t}+dJ_{t},\qquad X_{0}=x>0.
\end{equation}
\citet{Avanzi} studied the optimal dividend strategy
for maximizing the expected payments of all the future dividends
to the shareholders until the time of the ruin, that is
\begin{equation}\label{ValueFunction}
V(x):=\sup_{D\in\mathcal{D}}\mathbb{E}_{x}\left[\int_{0}^{\tau}e^{-\delta t}dD_{t}\right],
\end{equation}
with initial wealth $X_{0}=x$. 

They proved that the optimal strategy is a barrier strategy, that is, 
there exists an optimal barrier $b>0$, such that the optimal strategy is as follows. 
When the wealth process is below $b$, no dividend is paid out.
When the wealth process jumps above the barrier $b$ at the stopping time $\tau_{b}$,
$X_{\tau_{b}}-b$ is paid out as the dividends to the shareholders immediately
and the surplus drops to the level $b$. 
The value function $V(x)$ satisfies the equations:
\begin{equation}
V'(x)=1,\qquad\text{for any $x>b$},
\end{equation}
and for any $x<b$,
\begin{equation}\label{plugI}
-\rho V'(x)+\lambda\int_{0}^{\infty}[V(x+y)-V(x)]p(y)dy-\delta V(x)=0,
\end{equation}
with $V(0)=0$. For $x>b$, it is clear that $V(x)=x-b+V(b)$ and hence
we can plug this into \eqref{plugI} and obtain 
\begin{equation}
-\rho V'(x)-\lambda V(x)+\lambda\int_{0}^{b-x}V(x+y)p(y)dy
+\lambda\int_{b-x}^{\infty}(x+y-b+V(b))p(y)dy-\delta V(x)=0.
\end{equation}
By the smooth-fit condition, i.e., $V(x)$ is $C^{1}$ at $x=b$, we have 
\begin{equation}
-\rho-\lambda V(b)+\lambda\int_{0}^{\infty}yp(y)dy-\delta V(b)=0,
\end{equation}
which gives us the value of $V(b)$:
\begin{equation}\label{V(b)}
V(b)=\frac{\lambda\mathbb{E}[Y_{1}]-\rho}{\delta}.
\end{equation}

It was assumed in \citet{Avanzi} that $\rho<\lambda\mathbb{E}[Y_{1}]$. 
Under this condition, $\mathbb{P}_{x}(\tau=\infty)>0$
and $\mathbb{E}_{x}[\tau]=\infty$. 
They also assume that $\delta>0$.

In general, the optimal value $V(x)$ has no closed-form formulas.
But in some special cases, e.g., when $Y_{i}$ are exponentially
distributed, the optimal value $V(x)$ and the optimal barrier $b$
can be computed explicitly, see, e.g., \citet{Avanzi}.
From \citet{Avanzi} et al., when $p(y)=\nu e^{-\nu y}$, we have
\begin{equation}\label{V_exponential}
V(x)=\frac{\lambda}{\nu}\frac{e^{rx}-e^{sx}}{(\rho r+\delta)e^{rb}-(\rho s+\delta)e^{sb}}1_{[0,b]}(x)+(V(b)+x-b)1_{(b,\infty)}(x),
\end{equation}
where $r,s$ are the solutions of
\begin{equation}\label{quadratic_r_s}
\rho\xi^{2}+(\lambda+\delta-\nu\rho)\xi-\nu\delta=0,
\end{equation}
and the optimal $b$ is given by
\begin{equation}\label{b_exponential}
b
=\frac{1}{r-s}\log\left(\frac{s}{r}\frac{\rho s+\delta}{\rho r+\delta}\right).
\end{equation}
Without loss of generality, we can assume that $s>r$ and therefore from \eqref{quadratic_r_s}, we have 
\begin{equation}\label{r&s}
\begin{split}
&s=\frac{-(\lambda+\delta-\nu\rho)+\sqrt{(\lambda+\delta-\nu\rho)^{2}+4\rho\nu\delta}}{2\rho},
\\
&r=\frac{-(\lambda+\delta-\nu\rho)-\sqrt{(\lambda+\delta-\nu\rho)^{2}+4\rho\nu\delta}}{2\rho}.
\end{split}
\end{equation}

If we assume that 
\begin{equation}\label{lessAssump}
\rho\geq\lambda\mathbb{E}[Y_{1}], 
\end{equation}
then $\mathbb{P}_{x}(\tau<\infty)=1$, i.e., the ruin occurs with probability $1$.
Intuitively, it says that when you are certain that the company is going to get ruined, the optimal strategy
to maximize the dividend payments to shareholders is to give all the surplus
of the company to the shareholders immediately.
Therefore, for the finite-horizon case, under assumption \eqref{lessAssump},  we have the same conclusion. 
That is, for any time horizon $T>0$,
\begin{equation}
\sup_{D\in\mathcal{D}}\mathbb{E}_{x}\left[\int_{0}^{\tau\wedge T}e^{-\delta t}dD_{t}\right]=x.
\end{equation}

Notice that the assumption \eqref{lessAssump} is satisfied
when the innovation rate $\lambda\rightarrow 0$, or when the running cost $\rho\rightarrow\infty$. 
Therefore, these two asymptotics are trivial. 
We will study instead the $\lambda\rightarrow\infty$ asymptotics
and the $\rho\rightarrow  0$ asymptotics. 

Also notice that under the usual condition $\rho<\lambda\mathbb{E}[Y_{1}]$, 
$\mathbb{P}_{x}(\tau=\infty)>0$
and $\mathbb{E}_{x}[\tau]=\infty$. Therefore, if $\delta=0$, then
\begin{equation}\label{inftyEqn}
\sup_{D\in\mathcal{D}}\mathbb{E}_{x}\left[\int_{0}^{\tau}dD_{t}\right]=\infty.
\end{equation}
That is because we can always choose a constant dividend payment strategy
that is $D_{t}\equiv\hat{D}$, where $\hat{D}>0$ is a positive constant chosen sufficiently small
so that $\rho+\hat{D}<\lambda\mathbb{E}[Y_{1}]$. Then, let $\hat{\tau}$ be
the ruin time of this wealth process with $D_{t}\equiv\hat{D}$, we 
have $\hat{\tau}<\infty$ a.s. and $\mathbb{E}_{x}[\hat{\tau}]=\infty$.
Then, we have
\begin{equation}
\sup_{D\in\mathcal{D}}\mathbb{E}_{x}\left[\int_{0}^{\tau}dD_{t}\right]
\geq\hat{D}\mathbb{E}_{x}[\hat{\tau}]=\infty.
\end{equation}
Therefore, we expect that when $\delta\rightarrow 0$, 
\begin{equation}\label{finer:estimate}
\sup_{D\in\mathcal{D}}\mathbb{E}_{x}\left[\int_{0}^{\tau}e^{-\delta t}dD_{t}\right]
\rightarrow\infty.
\end{equation}

To obtain a finer estimate than \eqref{finer:estimate}
and understand how fast the term in \eqref{finer:estimate}
goes to infinity, we consider the limit
\begin{equation}
\lim_{\delta\to0}\sup_{D\in\mathcal{D}}\delta\mathbb{E}_{x}\left[\int_{0}^{\tau}e^{-\delta t}dD_{t}\right],
\end{equation}
specifically because of its asymptotic convergence and proper scaling. We will show that the above limit is finite and, therefore, we can determine how fast the value function approaches to infinity
as the discount rate $\delta\rightarrow 0$.

This is also of practical interest because the
value function for generally distributed $Y_{i}$  does not yield closed-form formula
and the asymptotic behavior is particularly useful in the low interest-rate environment
because a common choice of discount rate $\delta$ is by letting $\delta=r$, 
where $r>0$ is the risk-free rate. 

When $\delta=0$, we have seen already from \eqref{inftyEqn} that 
the value function is $\infty$. But we can also study the finite-horizon case
with a time horizon $T>0$. In the finite-horizon case, 
\begin{equation}
\sup_{D\in\mathcal{D}}\mathbb{E}_{x}\left[\int_{0}^{T\wedge\tau}dD_{t}\right]<\infty.
\end{equation}
But from \eqref{inftyEqn}, it is clear that
\begin{equation}
\sup_{D\in\mathcal{D}}\mathbb{E}_{x}\left[\int_{0}^{T\wedge\tau}dD_{t}\right]\rightarrow\infty,
\end{equation}
as $T\rightarrow\infty$. So we are interested to study how fast
this approaches to $\infty$ as $T\rightarrow\infty$.

When the discount rate $\delta\rightarrow\infty$, intuitively, it is clear
that the company should pay all the surplus as dividends to the shareholders
immediately because the cost of carry goes to infinity. 
When the time horizon $T\rightarrow0$, there is little time to accumulate new wealth
and what the company can pay to the shareholders is approximately the initial
wealth of the company.

To summarize, in this paper, we will focus on the following asymptotic regimes:
(i) Large innovation rate ($\lambda$) regime;
(ii) Small running cost ($\rho$) regime;
(iii) Small discount rate ($\delta$) regime; 
(iv) Large time horizon ($T$) regime;
(v) Large discount rate ($\delta$) regime; 
(vi) Small time horizon ($T$) regime.

It is well known that, see, e.g., \citet{Avanzi}, 
the optimal strategy for the optimal dividend problem
\begin{equation}
\sup_{D\in\mathcal{D}}\mathbb{E}_{x}\left[\int_{0}^{\tau}e^{-\delta t}dD_{t}\right]
\end{equation}
is a barrier strategy. 
But in practice, shareholders
prefer continuous dividend yield and most public companies
do not use barrier dividend strategies. 
We will show that in some asymptotic regimes,
continuous dividend yield strategy can be nearly optimal, even though not exactly optimal. 
More precisely, we define our nearly optimal dividend strategy $D^{M,\epsilon}$ as follows:

\begin{definition}\label{DDefn}
The strategy $D^{M,\epsilon}\in\mathcal{D}$ pays no dividend until $\tau_M:=\inf\{t\ge0\;:\;X_t\ge M\}$,  i.e., the first time that the process
jumps above $M$ before the ruin time. 
After $\tau_{M}$, a constant dividend yield $(1-\epsilon)(\lambda\mathbb{E}[Y_{1}]-\rho)$
is paid out to the shareholders.
\end{definition}

Let $\tau_{0}$ be the ruin time of the process $X_{t}=x-\rho t+J_{t}$
with no dividends.
Conditional on  $\tau_{M}< \tau_{0}$, a constant dividend yield $(1-\epsilon)(\lambda\mathbb{E}[Y_{1}]-\rho)$
is paid out to the shareholders until the ruin time $\tau^{\epsilon}$, which is defined as
\begin{equation}
\tau^{\epsilon}:=\inf\{t>0:X_{t}^{\epsilon}\leq 0\},
\end{equation}
where
\begin{equation}\label{eqn:X^epsilon}
dX_{t}^{\epsilon}=-\rho dt-(1-\epsilon)(\lambda\mathbb{E}[Y_{1}]-\rho)dt+dJ^{M}_{t},
\end{equation}
with $X_{0}^{\epsilon}=X_{\tau_M}$ and $J^{M}_t:=J_{\tau_{M}+t}-J_{\tau_{M}}$. 

Let $\tau^{M,\epsilon}$ be the ruin time of the process with dividend strategy $D^{M,\epsilon}$. Then, when $\tau_M>\tau_{0}$, we have $\tau^{M,\epsilon}=\tau_{0}$ and when $\tau_M<\tau_{0}$, we have $\tau^{M,\epsilon}=\tau_M+\tau^{\epsilon}$.

We will show that for small discount rate $\delta$ , large time horizon $T$, and  large innovation rate $\lambda$ regimes that for sufficiently large $M$ and sufficiently small $\epsilon$, $D^{M,\epsilon}$ is nearly optimal.

\section{Main Results}\label{sec:main}

We will focus on the following asymptotic regimes:
Large innovation rate ($\lambda$) regime (Section~\ref{sec:large:lambda});
Small running cost ($\rho$) regime (Section~\ref{sec:small:rho});
Small discount rate ($\delta$) regime (Section~\ref{sec:small:delta}); 
Large time horizon ($T$) regime (Section~\ref{sec:large:T});
Large discount rate ($\delta$) regime (Section~\ref{sec:large:delta}); 
Small time horizon ($T$) regime (Section~\ref{sec:small:T}).
We will use the notation $V(x;p)$ to emphasize the dependence of the value function $V$ on the parameter $p$; for example for small $\delta$ regime we use $V(x;\delta)$.


\subsection{Large Innovation Rate ($\lambda$) Regime}\label{sec:large:lambda}

Let us consider the large innovation rate ($\lambda$) asymptotics in this section.
This corresponds to the regime when the company has a large growth rate
and fast expansions. 
For the moment, let us assume that $Y_{i}$ are exponentially distributed, 
say $p(y)=\nu e^{-\nu y}$. 
Let us recall that for $x\ge b$ we have 
\begin{equation}
V(x;\lambda)=x-b+V(b;\lambda).
\end{equation}
Now as $\lambda\rightarrow\infty$, we can easily see from \eqref{r&s} that
$s\rightarrow 0$ and $r\rightarrow-\infty$
and by \eqref{b_exponential} $b\rightarrow 0$. The asymptotics of the value function is fully determined by
\begin{equation}
V(b;\lambda)\sim\frac{\lambda}{\nu\delta},\qquad\text{as $\lambda\rightarrow\infty$}.
\end{equation}

More generally, we have
\begin{equation}
V(b;\lambda)\sim\frac{\lambda\mathbb{E}[Y_{1}]}{\delta},
\qquad\text{as $\lambda\rightarrow\infty$}.
\end{equation}
The intuition is the following. As the innovation rate $\lambda\rightarrow\infty$, the ruin probability
will tend to zero. Let us assume that $X_{0}=x$. At any given time $t$, after a small time step $\Delta t$, 
the expected wealth increases to $x+(\lambda\mathbb{E}[Y_{1}]-\rho)\Delta t$, and then
at time $t+\Delta t$, you immediately pay the amount $(\lambda\mathbb{E}[Y_{1}]-\rho)\Delta t$ to the shareholders
and then you restart with wealth $x$ and continue the process. 
By letting $\Delta t\rightarrow 0$, we get
\begin{equation}
V(x;\lambda)\sim(\lambda\mathbb{E}[Y_{1}]-\rho)\int_{0}^{\infty}e^{-\delta t}dt=\frac{\lambda\mathbb{E}[Y_{1}]-\rho}{\delta}
\sim\frac{\lambda\mathbb{E}[Y_{1}]}{\delta}.
\end{equation}

\begin{theorem}\label{LargeLambdaThm}
(i) Assume the discount rate $\delta>0$. We have
\begin{equation}
\sup_{D\in\mathcal{D}}\mathbb{E}_{x}\left[\int_{0}^{\tau}e^{-\delta t}dD_{t}\right]
=\frac{\lambda\mathbb{E}[Y_{1}]}{\delta}(1+o(1)),
\quad
\text{as the innovation rate $\lambda\rightarrow\infty$}.
\end{equation}

(ii) Assume the discount rate $\delta>0$, for any time horizon $T>0$, we have
\begin{equation}
\sup_{D\in\mathcal{D}}\mathbb{E}_{x}\left[\int_{0}^{T\wedge\tau}e^{-\delta t}dD_{t}\right]
=\frac{\lambda\mathbb{E}[Y_{1}]}{\delta}(1-e^{-\delta T})(1+o(1)),
\quad\text{as the innovation rate $\lambda\rightarrow\infty$}.
\end{equation}

(iii) Assume the discount rate $\delta=0$, for any time horizon $T>0$, we have
\begin{equation}
\sup_{D\in\mathcal{D}}\mathbb{E}_{x}\left[\int_{0}^{T\wedge\tau}dD_{t}\right]
=\lambda\mathbb{E}[Y_{1}]T(1+o(1)),
\quad\text{as the innovation rate $\lambda\rightarrow\infty$}.
\end{equation}
\end{theorem}

We can see from Theorem~\ref{LargeLambdaThm}
that in the large innovation rate regime, i.e., with large $\lambda$,
the maximized expected discounted payment of all
the future dividends is large, which is
consistent with the intuition that innovation boosts the future value of the company and 
hence the payout to the shareholders. Indeed Theorem~\ref{LargeLambdaThm} implies
that the maximized expected discounted payment of all
the future dividends is of the order $\lambda$, i.e., linear in the innovation rate $\lambda$.

The optimal strategy is always the barrier strategy. 
But in practice, shareholders prefer continuous dividend yield.
Consider a dividend strategy $D^{M,\epsilon}$ from Definition \ref{DDefn} that pays continuous dividend
yield $(1-\epsilon)(\lambda\mathbb{E}[Y_{1}]-\rho)$ after the surplus process
hits above $M$.
We will show that the $D^{M,\epsilon}$ strategy is nearly optimal:

\begin{proposition}\label{epsilonOptimalLargeLambda}
Given any one of the cases in Theorem \ref{LargeLambdaThm}, 
for any $\varepsilon>0$, there exist some $M', \epsilon'>0$
such that 
for any $M>M'$ and $0<\epsilon<\epsilon'$, 
the $D^{M,\epsilon}$ strategy is $\varepsilon\lambda$-optimal, 
i.e., for any $\varepsilon>0$,  
and for any $\delta\geq 0$ and $T\in(0,\infty)$
or $\delta>0$ and $T=\infty$, 
there exist some $M', \epsilon'>0$
such that 
for any $M>M'$ and $0<\epsilon<\epsilon'$,
\begin{equation*}
\left|\sup_{D\in\mathcal{D}}\mathbb{E}_{x}\left[\int_{0}^{T\wedge\tau}e^{-\delta t}dD_{t}\right]
-\mathbb{E}_{x}\left[\int_{0}^{T\wedge\tau}e^{-\delta t}dD_{t}^{M,\epsilon}\right]\right|\leq\varepsilon\lambda,
\end{equation*}
for any sufficiently large $\lambda>0$.
\end{proposition}

We can see from Proposition~\ref{epsilonOptimalLargeLambda}
that $D^{M,\epsilon}$ strategy (defined in Definition~\ref{DDefn})
is nearly optimal in the large innovation rate ($\lambda$) regime, that is, the company does not pay out any dividend
until the surplus jumps above a high threshold $M$, 
and afterwards, a constant dividend yield $(1-\epsilon)(\lambda\mathbb{E}[Y_{1}]-\rho)$
is paid out to the shareholders. Managerial implication is that 
when the company has significant innovation power, it should focus on accumulating wealth by continuing
to innovate and deferring dividend payments to the shareholders until
the company has reached a considerably large scale.

\begin{remark}
It is also possible to have a discrete dividend strategy that is nearly optimal, see Remark \ref{Alternative}.
\end{remark}


\subsection{Small Running Cost ($\rho$) Regime}\label{sec:small:rho}

Let us consider the small running cost ($\rho$) asymptotics in this section.
We will see that when the running cost $\rho=0$, for any $x>0$, 
\begin{equation}
V(x;\rho)=x+\frac{\lambda\mathbb{E}[Y_{1}]}{\delta}.
\end{equation}
The intuition is that since when the running cost $\rho\rightarrow 0$, 
the probability of ruin is negligible, and if the discount rate $\delta>0$, it is optimal to give almost all the surplus
as dividends to the shareholders immediately rather than holding it. Then, each time the surplus increases, we also pay excess surplus as dividend. More precisely, for $\epsilon>0$, let $D^\epsilon$ be the strategy which pays dividend $x-\epsilon$ at the beginning and any surplus above $\epsilon$ thereafter.
At time $0$, with initial wealth $x$, you give $x-\epsilon$ dollars
as amount of dividends to the shareholders. Then at time of the $n$th jump of the process $J$, $\tau^{(n)}$, the  wealth grows to $Y_{n}+\epsilon$
and then you give  $Y_{n}$ as dividend. Therefore, when the running cost $\rho=0$,
\begin{equation}\label{rho=0}
V(x;\rho)=x+\int_{0}^{\infty}e^{-\delta t}\lambda\mathbb{E}[Y_{1}]dt
=x+\frac{\lambda\mathbb{E}[Y_{1}]}{\delta}.
\end{equation}

To have a more rigorous proof, notice that for $x>b$, 
$V(x;\rho)=x-b+V(b;\rho)$. Recall from \eqref{V(b)} that $V(b;\rho)=\frac{\lambda\mathbb{E}[Y_{1}]-\rho}{\delta}$. 
Hence when the running cost $\rho=0$, we have \eqref{rho=0}.
In fact, the optimal strategy for the finite horizon case should be the same.

\begin{theorem}\label{SmallRhoThm}
Assume that the running cost $\rho=0$.
 
(i) For  any discount rate $\delta>0$, we have
\begin{equation}\label{SmallRhoT=infty}
\sup_{D\in\mathcal{D}}\mathbb{E}_{x}\left[\int_{0}^{\tau}e^{-\delta t}dD_{t}\right]
=x+\frac{\lambda\mathbb{E}[Y_{1}]}{\delta}.
\end{equation}

(ii) For any discount rate $\delta>0$ and time horizon $T>0$, we have
\begin{equation}\label{SmallRhoTfinitedelta>0}
\sup_{D\in\mathcal{D}}\mathbb{E}_{x}\left[\int_{0}^{T\wedge\tau}e^{-\delta t}dD_{t}\right]
=x+\frac{\lambda\mathbb{E}[Y_{1}]}{\delta}(1-e^{-\delta T}).
\end{equation}

(iii) When the discount rate $\delta=0$, for any time horizon $T>0$, we have
\begin{equation}\label{SmallRhoTfinitedelta=0}
\sup_{D\in\mathcal{D}}\mathbb{E}_{x}\left[\int_{0}^{T\wedge\tau}dD_{t}\right]
=x+\lambda\mathbb{E}[Y_{1}]T.
\end{equation}
\end{theorem}

\begin{proposition}\label{epsilonOptimalSmallRho}
Assume that the running cost $\rho=0$.
Given any of the cases in Theorem \ref{SmallRhoThm}, 
for any $\varepsilon>0$, there exists an $\epsilon>0$
such that the strategy $D^{\epsilon}$ is $\varepsilon$-optimal, 
i.e., 
for any $\varepsilon>0$
and for any $\delta\geq 0$, $T\in(0,\infty)$
or $\delta>0$, $T=\infty$, there exists an $\epsilon>0$
such that
\begin{equation*}
\left|\sup_{D\in\mathcal{D}}\mathbb{E}_{x}\left[\int_{0}^{T\wedge\tau}e^{-\delta t}dD_{t}\right]
-\mathbb{E}_{x}\left[\int_{0}^{T\wedge\tau}e^{-\delta t}dD_{t}^{\epsilon}\right]\right|\leq\varepsilon.
\end{equation*}
\end{proposition}

We can see from Proposition~\ref{epsilonOptimalSmallRho}
that $D^\epsilon$ is nearly optimal, that is, to pay dividend $x-\epsilon$ at the beginning and any surplus above $\epsilon$ thereafter.
The intuition is that when the running cost $\rho$ is small, so is the ruin risk. 
Since the future value of the payment is always smaller than the present value
due to the discount rate, it is nearly optimal to pay as much dividend as possible to the shareholders
at time zero.

In practice, the running cost $\rho$ should always be positive,
even though it can be small. Therefore, the discussions
for the $\rho=0$ case in Theorem~\ref{SmallRhoThm} and Proposition~\ref{epsilonOptimalSmallRho}
serve as a first-order approximation when $\rho>0$ is small, instead
of describing a real world scenario.
Indeed, for the special case when $p(y)=\nu e^{-\nu y}$, we can even find the second-order approximation as 
the running cost $\rho\rightarrow 0$.
Let us recall that for $x>b$,  $V(x;\rho)=x-b+V(b;\rho)$, where $b$ is given by \eqref{b_exponential}. From \eqref{r&s}, it is easy that see that $r\sim\frac{-(\lambda+\delta)}{\rho}$
and $s\sim\frac{\nu\delta}{\lambda+\delta}$. 
This implies that 
\begin{equation}
b\sim\frac{\rho}{-(\lambda+\delta)}\log\left(\frac{\nu\delta}{\lambda+\delta}\frac{-\rho}{\lambda+\delta}\frac{\delta}{-\lambda}\right)
\sim\frac{-\rho\log\rho}{\lambda+\delta},
\end{equation}
as the running cost $\rho\rightarrow 0$. Hence, we conclude that for any $x>0$,
\begin{equation}
V(x;\rho)\sim x+\frac{\lambda}{\nu\delta}+\frac{\rho\log\rho}{\lambda+\delta},
\end{equation}
as the running cost $\rho\rightarrow 0$.


\subsection{Small Discount Rate ($\delta$) Regime}\label{sec:small:delta}

Let us consider the small discount rate ($\delta$) asymptotics
in this section. This is practically relevant when the interest rate
is low, which is a new environment, e.g.,
after the 2008 financial crisis in the United States.
Recall that $\lambda\mathbb{E}[Y_{1}]>\rho$ so that
$\mathbb{P}_{x}(\tau=\infty)>0$ and $\mathbb{E}_{x}[\tau]=\infty$.
Therefore, by considering a constant dividend yield strategy, it easily follows that
\begin{equation}
V(x;\delta)=\sup_{D\in\mathcal{D}}\mathbb{E}_{x}\left[\int_{0}^{\tau}e^{-\delta t}dD_{t}\right]\rightarrow\infty,
\qquad
\text{as $\delta\rightarrow 0$}.
\end{equation}
We are interested to see how fast it goes to $\infty$ as $\delta\rightarrow 0$.

To get some intuitions, let us first consider the case when $p(y)=\nu e^{-\nu y}$
so that there are explicit formulas for the optimal value function
and the optimal barrier. Let us recall that for $x\le b$,
\begin{equation}
V(x;\delta)=\frac{\lambda}{\nu}\frac{e^{rx}-e^{sx}}{(\rho r+\delta)e^{rb}-(\rho s+\delta)e^{sb}},
\end{equation}
where $r,s$ are given by \eqref{r&s}
and the optimal $b$ is given by \eqref{b_exponential}.
Therefore, as $\delta\rightarrow 0$, we have
$s\rightarrow 0$ and $r\rightarrow\nu-\frac{\lambda}{\rho}$ and $b\to\infty$.

By the definition of the optimal $b$, we have
\begin{equation}
\frac{r}{s}(\rho r+\delta)e^{br}=(\rho s+\delta)e^{bs}.
\end{equation}
This implies that for $x\le b$
\begin{equation}
V(x;\delta)=\frac{\lambda}{\nu}\frac{e^{rx}-e^{sx}}{\left(1-\frac{r}{s}\right)(\rho r+\delta)e^{br}}.
\end{equation}
Notice that
\begin{equation}
e^{br}=e^{\frac{r}{r-s}\log\left(\frac{s^{2}(\nu-r)}{r^{2}(\nu-s)}\right)}
=\left(\frac{s^{2}(\nu-r)}{r^{2}(\nu-s)}\right)^{\frac{r}{r-s}}
\sim\frac{\frac{\lambda}{\rho}s^{2}}{\nu(\nu-\frac{\lambda}{\rho})^{2}},
\end{equation}
as $\delta\rightarrow 0$. 
Therefore, we have
\begin{equation}
V(x;\delta)\sim\frac{\lambda(1-e^{(\nu-\frac{\lambda}{\rho})x})}{r\nu(\rho\nu-\lambda)\frac{\lambda}{\rho}}
\nu\left(\nu-\frac{\lambda}{\rho}\right)^{2}\frac{1}{s},
\end{equation}
as $\delta\rightarrow 0$.

Note that $s\sim\frac{\nu}{\lambda-\nu\rho}\delta$ as $\delta\rightarrow 0$. 
Therefore, we conclude that
\begin{equation}
\lim_{\delta\rightarrow 0}\delta V(x;\delta)
=\frac{\lambda-\nu\rho}{\nu}(1-e^{(\nu-\frac{\lambda}{\rho})x}).
\end{equation}

For generally distributed $Y_{i}$, we provide some heuristic arguments. 
Notice that for the optimization problem 
\begin{equation}
V(x;\delta)=\sup_{D\in\mathcal{D}}\mathbb{E}_{x}\left[\int_{0}^{\tau}e^{-\delta t}dD_{t}\right],
\end{equation}
the optimal strategy is a barrier strategy, that is,
\begin{equation}
V'(x;\delta)=1,\qquad\text{for any $x>b$},
\end{equation}
and for any $x<b$,
\begin{equation}
-\rho V'(x;\delta)+\lambda\int_{0}^{\infty}[V(x+y;\delta)-V(x;\delta)]p(y)dy-\delta V(x;\delta)=0,
\end{equation}
with $V(0;\delta)=0$. 

As $\delta\rightarrow 0$, the optimal $b\rightarrow\infty$. 
Therefore, for fixed $x$, we have $x/b\rightarrow 0$ and $V(x;\delta)$ roughly satisfies the equation:
\begin{equation}
-\rho V'(x;\delta)+\lambda\int_{0}^{\infty}[V(x+y;\delta)-V(x;\delta)]p(y)dy=0,
\end{equation}
with $V(0;\delta)=0$, which yields that
\begin{equation}
V(x;\delta)\sim c(1-e^{-\alpha x}),
\end{equation}
for some positive constant $c$, where $\alpha$ is the unique positive solution to the equation:
\begin{equation}
\rho\alpha+\lambda\int_{0}^{\infty}[e^{-\alpha y}-1]p(y)dy=0.
\end{equation}
Next, let us determine the positive constant $c$. 
Recall that by \eqref{V(b)}, for any $x>b$,
\begin{equation}
V(x;\delta)=\frac{\lambda\mathbb{E}[Y_{1}]-\rho}{\delta}+(x-b)+o(1).
\end{equation}
This implies that for $x$ fixed and large, $V(x;\delta)\delta\sim\lambda\mathbb{E}[Y_{1}]-\rho$ as $\delta\rightarrow 0$.
Hence, we have $c=\frac{\lambda\mathbb{E}[Y_{1}]-\rho}{\delta}$ and
\begin{equation}
\lim_{\delta\rightarrow 0}\delta V(x;\delta)=(\lambda\mathbb{E}[Y_{1}]-\rho)(1-e^{-\alpha x}).
\end{equation}

Indeed, we can prove it rigorously:

\begin{theorem}\label{SmallDeltaThm}
We have the following asymptotic result for small discount rate $\delta$:
\begin{equation}
\sup_{D\in\mathcal{D}}\mathbb{E}_{x}\left[\int_{0}^{\tau}e^{-\delta t}dD_{t}\right]
=\frac{\lambda\mathbb{E}[Y_{1}]-\rho}{\delta}(1-e^{-\alpha x})(1+o(1)),
\qquad
\text{as the discount rate $\delta\rightarrow 0$}.
\end{equation}
\end{theorem}

We can see from Theorem~\ref{SmallDeltaThm}
that in the low interest rate environment, i.e., the small discount rate $\delta$,
the maximized expected discounted payment of all
the future dividends is large, which is
consistent with the conventional wisdom that low interest rate environment
boosts the asset values. Indeed Theorem~\ref{SmallDeltaThm} implies
that the maximized expected discounted payment of all
the future dividends is of the order $1/\delta$, i.e., the reciprocal of the discount rate.

\begin{proposition}\label{epsilonOptimalSmallDelta}
For any $\varepsilon>0$, let $\delta>0$ be such that
\begin{equation}
\left|\delta\sup_{D\in\mathcal{D}}\mathbb{E}_{x}\left[\int_{0}^{\tau}e^{-\delta t}dD_{t}\right]-(\lambda\mathbb{E}[Y_{1}]-\rho)(1-e^{-\alpha x})\right|
\leq\varepsilon.
\end{equation}
Then, there exist some $M', \epsilon'>0$
such that 
for any $M>M'$ and $0<\epsilon<\epsilon'$, 
the $D^{M,\epsilon}$ strategy is an $\frac{\varepsilon}{\delta}$-optimal strategy, i.e., 
\begin{equation}
\left|\delta\mathbb{E}_{x}\left[\int_{0}^{\tau}e^{-\delta t}dD^{M,\epsilon}_{t}\right]-(\lambda\mathbb{E}[Y_{1}]-\rho)(1-e^{-\alpha x})\right|
\leq\varepsilon.
\end{equation}
\end{proposition}

We can see from Proposition~\ref{epsilonOptimalSmallDelta}
that $D^{M,\epsilon}$ strategy (defined in Definition~\ref{DDefn})
is nearly optimal, that is, the company does not pay out any dividend
until the surplus jumps above a high threshold $M$, 
and afterwards, a constant dividend yield $(1-\epsilon)(\lambda\mathbb{E}[Y_{1}]-\rho)$
is paid out to the shareholders. Managerial implication is that in the low interest rate environment, 
the company should focus on accumulating wealth and deferring dividend payments until
the company has reached a considerably large scale.

\subsection{Large Time Horizon ($T$) Regime}\label{sec:large:T}

Let us consider the large time horizon ($T$) asymptotics in this section. 
Denote the value function as
\begin{equation}
V(x;T):=\sup_{D\in\mathcal{D}}\mathbb{E}_{x}\left[\int_{0}^{T\wedge\tau}e^{-\delta t}dD_{t}\right].
\end{equation}
Let us differentiate two cases: $\delta>0$ and $\delta=0$.

When $\delta>0$, intuitively, it is clear that the finite-horizon value function
will converge to the infinite-horizon value function as time horizon $T\rightarrow\infty$:
\begin{equation}
V(x;T)=\sup_{D\in\mathcal{D}}\mathbb{E}_{x}\left[\int_{0}^{T\wedge\tau}e^{-\delta t}dD_{t}\right]
\rightarrow
\sup_{D\in\mathcal{D}}\mathbb{E}_{x}\left[\int_{0}^{\tau}e^{-\delta t}dD_{t}\right].
\end{equation}
Indeed, we can give an upper bound
on the speed of the convergence and show rigorously the following result:

\begin{theorem}\label{LargeT}
Assume $\delta>0$, we have
\begin{equation}
\left|V(x;T)-\sup_{D\in\mathcal{D}}\mathbb{E}_{x}\left[\int_{0}^{\tau}e^{-\delta t}dD_{t}\right]\right|
\leq(x+\lambda\mathbb{E}[Y_{1}]T)e^{-\delta T}+\frac{\lambda\mathbb{E}[Y_{1}]}{\delta}e^{-\delta T}.
\end{equation}
\end{theorem}

Next, let us consider the $\delta=0$ case, i.e.,
\begin{equation}
V(x;T)=\sup_{D\in\mathcal{D}}\mathbb{E}_{x}\left[\int_{0}^{T\wedge\tau}dD_{t}\right].
\end{equation}
Then, we expect that
\begin{equation}
\lim_{T\rightarrow\infty}\frac{V(x;T)}{T}=(\lambda\mathbb{E}[Y_{1}]-\rho)(1-e^{-\alpha x}).
\end{equation}
We shall show this result rigorously later. 
Before we proceed, let us give some heuristic arguments
and gain some intuition behind this result.
Notice that when the discount rate $\delta=0$, the present value of the future dividends
does not decay as the payment time evolves. Moreover, 
$\lambda\mathbb{E}[Y_{1}]>\rho$, so there is an upward drift
and if the company holds the wealth rather than pay the dividends, there will
be less chance that the company is going to get ruined. On the other hand,
we have already seen that the value of the dividends do not decay over time
because $\delta=0$, therefore, the optimal strategy
is not to pay any dividends until the time of the maturity $T$
if by that time the company is not ruined.
As the time horizon $T\rightarrow\infty$, the probability that the company is not ruined
is the ultimate survival probability, that is, $1-e^{-\alpha x}$.
On the other hand, ignoring ruin probability, 
the wealth of the company right before time $T$ is given by \eqref{dual:risk:model}. {\color{blue} By taking expectation, we obtain
\begin{equation}
\mathbb{E}_{x}[X_{t}]=x-\rho t + \mathbb{E}[J_t],~ t<T.
\end{equation}
By monotone convergence theorem, 
\begin{equation}
\mathbb{E}_{x}[X_{T-}]=x-\rho t + \mathbb{E}[J_T]= x+(\lambda\mathbb{E}[Y_{1}]-\rho)T.
\end{equation}
where $X_{T-}$ is the left limit of the process $X_{t}$ at $t\uparrow T$.}
So for large $T$, on average, $\lambda\mathbb{E}[Y_{1}]-\rho$ is the rate of the growth of the company. 
Therefore, we expect to get
\begin{equation}\label{lim_then_sup}
\lim_{T\rightarrow\infty}\frac{V(x;T)}{T}=(\lambda\mathbb{E}[Y_{1}]-\rho)(1-e^{-\alpha x}).
\end{equation}

The rigorous result is as follows.
\begin{theorem}\label{LargeT_delta=0}
As the time horizon $T\rightarrow\infty$, we have
\begin{equation}
\sup_{D\in\mathcal{D}}\mathbb{E}_{x}\left[\int_{0}^{T\wedge\tau}dD_{t}\right]
=(\lambda\mathbb{E}[Y_{1}]-\rho)\left[(1-e^{-\alpha x})T+\frac{e^{-\alpha x}x}{\rho-\lambda\int_{0}^{\infty}e^{-\alpha y}yp(y)dy}+o(1)\right].
\end{equation}
\end{theorem}

The optimal strategy when $\delta=0$ is to withhold any dividend
until $T\wedge\tau$. But in practice, that is not very realistic. 
Indeed, we will show that the $D^{M,\epsilon}$ strategy that pays continuous
dividend yield after the surplus reaches above the level $M$ is nearly optimal:

\begin{proposition}\label{epsilonOptimalLargeT}
For any $\varepsilon>0$, let $T>0$ be such that
\begin{equation}
\left|\frac{1}{T}\sup_{D\in\mathcal{D}}\mathbb{E}_{x}\left[\int_{0}^{\tau\wedge T}dD_{t}\right]-(\lambda\mathbb{E}[Y_{1}]-\rho)(1-e^{-\alpha x})\right|
\leq\varepsilon.
\end{equation}
Then, there exist some $M', \epsilon'>0$
such that 
for any $M>M'$ and $0<\epsilon<\epsilon'$,  
the $D^{M,\epsilon}$ strategy is an $\varepsilon T$-optimal strategy, i.e., 
\begin{equation}
\left|\frac{1}{T}\mathbb{E}_{x}\left[\int_{0}^{\tau\wedge T}dD^{M,\epsilon}_{t}\right]-(\lambda\mathbb{E}[Y_{1}]-\rho)(1-e^{-\alpha x})\right|
\leq\varepsilon.
\end{equation}
\end{proposition}

We can see from Proposition~\ref{epsilonOptimalLargeT}
that $D^{M,\epsilon}$ strategy (defined in Definition~\ref{DDefn})
is nearly optimal in the zero discount rate $\delta=0$
and large time horizon $(T)$ regime, that is, the company does not pay out any dividend
until the surplus jumps above a high threshold $M$, 
and afterwards, a constant dividend yield $(1-\epsilon)(\lambda\mathbb{E}[Y_{1}]-\rho)$
is paid out to the shareholders. Managerial implication is that in the low rate environment (small $\delta$)
with a long-term view (large $T$), 
the company should focus on accumulating wealth and deferring dividend payments until
the company has reached a considerably large scale. This is consistent 
with our conclusion from Proposition~\ref{epsilonOptimalSmallDelta}.



\subsection{Large Discount Rate ($\delta$) Regime}\label{sec:large:delta}

Let us consider the large discount rate ($\delta$) asymptotics in this section.
When the discount rate $\delta\rightarrow\infty$, intuitively, it becomes clear
that if the company waits, the present value of the future dividends
will be virtually zero because of the extreme discount factor. 
Therefore, it is easy to see that the optimal strategy is to give the surplus as the dividends
to the shareholders sooner rather than later. Intuitively, one might guess
that all the surplus should be given to the shareholders as the dividends at time zero and
thus 
\begin{equation}
\sup_{D\in\mathcal{D}}\mathbb{E}_{x}\left[\int_{0}^{T\wedge\tau}e^{-\delta t}dD_{t}\right]\sim x,
\qquad\text{as the discount rate $\delta\rightarrow\infty$}. 
\end{equation}
We will show later that this indeed is true. 
Moreover, we can obtain the second order approximations as the discount rate $\delta\rightarrow\infty$ 
when $Y_{i}$ are exponentially distributed.

Let us consider the special case when $p(y)=\nu e^{-\nu y}$. 
Then by \eqref{V(b)}, for any $x>b$, we have
\begin{equation}
V(x;\delta)=x-b+\frac{\lambda\mathbb{E}[Y_{1}]-\rho}{\delta}
=x-b+\frac{\frac{\lambda}{\nu}-\rho}{\delta}.
\end{equation}
The optimal barrier $b$ is given by \eqref{b_exponential}.
Then from \eqref{r&s}, it is easy to see that $r\sim-\frac{\delta}{\rho}$ as $\delta\rightarrow\infty$
and $s\sim\nu$ as $\delta\rightarrow\infty$.
Furthermore, we can compute that
\begin{align}
\rho r+\delta
&=\frac{-(\lambda+\delta-\nu\rho)}{2}-\frac{\sqrt{(\lambda+\delta-\nu\rho)^{2}+4\rho\nu\delta}}{2}+\delta
\\
&=\frac{\delta-\lambda+\nu\rho-\sqrt{(\delta-\lambda+\nu\rho)^{2}+4\delta\lambda}}{2}
\nonumber
\\
&=\frac{-2\delta\lambda}{\delta-\lambda+\nu\rho+\sqrt{(\delta-\lambda+\nu\rho)^{2}+4\delta\lambda}}
\sim-\lambda,
\nonumber
\end{align}
as the discount rate $\delta\rightarrow\infty$. Therefore, plugging into \eqref{b_exponential}, we get
\begin{equation}
b\sim\frac{\rho\log(\frac{\lambda}{\nu\rho})}{\delta},
\qquad\text{as the discount rate $\delta\rightarrow\infty$}.
\end{equation}
Hence, we conclude that for any $x>0$, 
\begin{equation}
V(x;\delta)=x-b+\frac{\frac{\lambda}{\nu}-\rho}{\delta}
\sim x+\frac{1}{\delta}\left(\frac{\lambda}{\nu}-\rho-\rho\log\left(\frac{\lambda}{\nu\rho}\right)\right),
\end{equation}
as the discount rate $\delta\rightarrow\infty$.

We can indeed prove the first order approximation for generally distributed $Y_{i}$ rigorously 
and have the following result:

\begin{theorem}\label{LargeDeltaThm}
\begin{equation}
\sup_{D\in\mathcal{D}}\mathbb{E}_{x}\left[\int_{0}^{\tau}e^{-\delta t}dD_{t}\right]
=x(1+o(1)),\qquad\text{as the discount rate $\delta\rightarrow\infty$}.
\end{equation}
\end{theorem}

When the discount rate $\delta$ goes to infinity, intuitively, it is clear
that the company should pay all the surplus as dividends to the shareholders
immediately because the cost of carry goes to infinity. 
Indeed, we can see from Theorem~\ref{LargeDeltaThm}
that a nearly optimal strategy is to pay out the amount of the initial wealth $x$
at time zero. Therefore in the high discount rate (e.g., interest rate) environment, 
to maximize the shareholders value, the company should issue the dividend payment
immediately instead of focusing on the growth.


\subsection{Small Time Horizon ($T$) Regime}\label{sec:small:T}

Let us consider the small time horizon ($T$) asymptotics in this section.
When the time horizon $T\rightarrow 0$, there is little time for the company
to accumulate new wealth and thus all what the company
can pay to the shareholders is roughly $x$. 
Within the small time interval $[0,T]$, for sufficiently small, 
the ruin probability is zero because ruin will not occur until after
the time $\frac{x}{\rho}$. So on this short time interval,
there is no ruin risk. For small time horizon $T$, it is easy
to compute that $\mathbb{E}_{x}[X_{T}]=x+(\lambda\mathbb{E}[Y_{1}]-\rho)T$. 
For any dividend, the difference between paying the dividend at time $0$
and at the time $T$ is only a discount factor at most $e^{-\delta T}$. 
Since there is no ruin risk, the initial wealth $x$ should be given upfront
as the dividends to shareholders. $(\lambda\mathbb{E}[Y_{1}]-\rho)T$ is
of order $O(T)$, and thus the impact of the discount factor
on this amount of new wealth is of order $O(T^{2})$ which is negligible. 
Therefore, we have the following result:

\begin{theorem}\label{SmallTThm}
Assume that $\lambda\mathbb{E}[Y_{1}]>\rho$. As the time horizon $T\rightarrow 0$, we have
\begin{equation}
\sup_{D\in\mathcal{D}}\mathbb{E}_{x}\left[\int_{0}^{\tau\wedge T}e^{-\delta t}dD_{t}\right]
=x+(\lambda\mathbb{E}[Y_{1}]-\rho)T+O(T^{2}).
\end{equation}
\end{theorem}

When the time horizon $T$ goes to zero, there is little time to accumulate new wealth
and what the company can pay to the shareholders is approximately the initial
wealth $x$ of the company. Indeed, we can see from Theorem~\ref{SmallTThm}
that a nearly optimal strategy is to pay out the amount of the initial wealth $x$ at time zero.
Therefore with a short-term view (i.e., small time horizon $T$), 
to maximize the shareholders value, the company should issue the dividend payment
immediately instead of focusing on the growth.


\section{Proofs of the Main Results}\label{sec:proofs}

\subsection{Large Innovation Rate ($\lambda$) Regime}

\begin{proof}[Proof of Theorem \ref{LargeLambdaThm} and Proposition \ref{epsilonOptimalLargeLambda}]
(i) First, let us prove the upper bound. 
Notice that the optimal strategy is the barrier strategy with
\begin{equation}
V(x;\lambda)=x-b+V(b;\lambda),\qquad\text{for $x>b$},
\end{equation}
and $V'(x;\lambda)>1$ for $x<b$ and $V(0;\lambda)=0$. Therefore, it is easy to
see that for $x<b$, we have $V(x;\lambda)\leq x-b+V(b;\lambda)$. 
On the other hand by \eqref{V(b)}, $V(b;\lambda)=\frac{\lambda\mathbb{E}[Y_{1}]-\rho}{\delta}$.
Therefore, for any $x$,
\begin{equation}
V(x;\lambda)\leq x-b+\frac{\lambda\mathbb{E}[Y_{1}]-\rho}{\delta}
\leq x+\frac{\lambda\mathbb{E}[Y_{1}]-\rho}{\delta}.
\end{equation}
This gives us the upper bound.

Next, let us prove the lower bound. For any $M>0$ and $\epsilon>0$, let us recall that the definition of the dividend strategy $D^{M,\epsilon}$ in Definition \ref{DDefn}:
no dividend is paid out until the first time that the process
jumps above $M$ and then it pays dividend with continuous rate $(1-\epsilon)(\lambda\mathbb{E}[Y_1]-\rho)$, and also recall the definitions of $\tau_{M}$, $X_{t}^{\epsilon}$, $\tau^{\epsilon}$, $\tau^{M,\epsilon}$, and $\tau_{0}$.
Then, 
\begin{align}
\sup_{D\in\mathcal{D}}\mathbb{E}_{x}\left[\int_{0}^{\tau}e^{-\delta t}dD_{t}\right]
&\geq\mathbb{E}_{x}\left[\int_{0}^{\tau^{M,\epsilon}}e^{-\delta t}dD_{t}^{M,\epsilon}\right]
\\
&\geq\mathbb{E}_{x}\left[e^{-\delta\tau_{M}}1_{\tau_{M}<\tau_{0}}\right]
(1-\epsilon)(\lambda\mathbb{E}[Y_{1}]-\rho)\mathbb{E}\left[\int_{0}^{\tau^{\epsilon}}e^{-\delta t}dt\middle|X_{0}^{\epsilon}=M\right]
\nonumber
\\
&=\mathbb{E}_{x}\left[e^{-\delta\tau_{M}}1_{\tau_{M}<\tau_{0}}\right](1-\epsilon)\frac{(\lambda\mathbb{E}[Y_{1}]-\rho)}{\delta}
\left(1-\mathbb{E}\left[e^{-\delta\tau^{\epsilon}}\middle|X_{0}^{\epsilon}=M\right]\right)
\nonumber
\\
&=\mathbb{E}_{x}\left[e^{-\delta\tau_{M}}1_{\tau_{M}<\tau_{0}}\right](1-\epsilon)\frac{(\lambda\mathbb{E}[Y_{1}]-\rho)}{\delta}
\left(1-e^{-\beta^{\epsilon,\lambda}M}\right),
\nonumber
\end{align}
where $\beta^{\epsilon,\lambda}$  is the unique positive value
that satisfies the equation
\begin{equation}
\rho\beta^{\epsilon,\lambda}+(1-\epsilon)(\lambda\mathbb{E}[Y_{1}]-\rho)\beta^{\epsilon,\lambda}
+\lambda\int_{0}^{\infty}[e^{-\beta^{\epsilon,\lambda}y}-1]p(y)dy-\delta=0,
\end{equation}
It is easy to check that $\beta^{\epsilon,\lambda}\rightarrow\beta^{\epsilon,\infty}$,
where $\beta^{\epsilon,\infty}$ is the unique positive value that satisfies:
\begin{equation}
(1-\epsilon)\mathbb{E}[Y_{1}]\beta^{\epsilon,\infty}
+\int_{0}^{\infty}[e^{-\beta^{\epsilon,\infty}y}-1]p(y)dy=0.
\end{equation}
Since $\tau_{M}\to0$ when $\lambda\to\infty$,  it follows from bounded convergence theorem that
\begin{equation}
\lim_{\lambda\rightarrow\infty}\mathbb{E}_{x}\left[e^{-\delta\tau_{M}}1_{\tau_{M}<\tau_{0}}\right]
=1.
\end{equation}
Therefore, 
\begin{equation}
\liminf_{\lambda\rightarrow\infty}\frac{1}{\lambda}
\sup_{D\in\mathcal{D}}\mathbb{E}_{x}\left[\int_{0}^{\tau}e^{-\delta t}dD_{t}\right]
\geq(1-\epsilon)\frac{\mathbb{E}[Y_{1}]}{\delta}\left(1-e^{-\beta^{\epsilon,\infty}M}\right).
\end{equation}
Finally, letting $M\rightarrow\infty$ first  and then $\epsilon\rightarrow 0$, we get
\begin{equation}
\liminf_{\lambda\rightarrow\infty}\frac{1}{\lambda}
\sup_{D\in\mathcal{D}}\mathbb{E}_{x}\left[\int_{0}^{\tau}e^{-\delta t}dD_{t}\right]
\geq\frac{\mathbb{E}[Y_{1}]}{\delta}.
\end{equation}

(ii) Assume $\delta>0$ and $T>0$. 
Let us first prove the upper bound.
Let 
\begin{equation}
V(x,t;\lambda):=\sup_{D\in\mathcal{D}}\mathbb{E}_{x}\left[\int_{t}^{T\wedge\tau}e^{-\delta t}dD_{t}\right].
\end{equation}
Then, $V(x,t)$ satisfies the equation:
\begin{equation}
\max\left\{
\frac{\partial V}{\partial t}
-\rho\frac{\partial V}{\partial x}
+\lambda\int_{0}^{\infty}[V(x+y,t;\lambda)-V(x,t;\lambda)]p(y)dy-\delta V,-\frac{\partial V}{\partial x}+1
\right\}=0,
\end{equation}
with terminal condition $V(x,T;\lambda)=x$. Now define the function $U_1(x,t):=x+\frac{\lambda\mathbb{E}[Y_1]-\rho}{\delta}(1-e^{-\delta(T-t)})$ and consider an arbitrage dividend strategy $D\in\mathcal{D}$. Recall that $dX_t=\rho dt +dJ_t-dD_t$. Thus, by change of variable formula for processes of bounded variation (see, e.g., \cite[Theorem~II.31]{Protter}), we  obtain
\begin{align}
U_{1}(x,0)&=\mathbb{E}_{x}[e^{-\delta(T\wedge\tau)}U_{1}(X_{T\wedge\tau},T\wedge\tau)]
+\mathbb{E}_{x}\Bigg[\int_{0}^{T\wedge\tau}e^{-\delta s}\bigg(-\frac{\partial U_{1}}{\partial t}(X_{s},s)+\rho \frac{\partial U_{1}}{\partial x}(X_{s},s)\\
&-\lambda\int_{0}^{\infty}[U_{1}(X_{s}+y,s)-U_{1}(X_{s},s)]p(y)dy
+\delta U_{1}(X_{s},s)\bigg)ds\Bigg]
\nonumber
\\
&\qquad
+\mathbb{E}_{x}\left[\int_{0}^{T\wedge\tau}e^{-\delta s}\frac{\partial U_{1}}{\partial x}(X_{s},s)dD_{s}\right]
\nonumber
\\
&-\mathbb{E}_{x}\left[\sum_{s\leq T\wedge\tau}e^{-\delta s}\left[U_{1}(X_{s+},s)-U_{1}(X_{s},s)+\frac{\partial U_{1}}{\partial x}(X_{s},s)\Delta D_{s}\right]\right].
\nonumber
\end{align}
By direct calculation, the Riemann integral in the above and the term 
\[
\mathbb{E}_{x}[e^{-\delta(T\wedge\tau)}U_{1}(X_{T\wedge\tau},T\wedge\tau)]
\]
 are non-negative. Also for the last term, we have 
\[
U_{1}(X_{s+},s)-U_{1}(X_{s},s)-\frac{\partial U_{1}}{\partial x}(X_{s},s)\Delta D_{s}=0.
\]
Therefore, we can write
\[
U_{1}(x,0)\ge \mathbb{E}_{x}\left[\int_{0}^{T\wedge\tau}e^{-\delta s}dD_{s}\right]
\]
Taking supremum over $D\in\mathcal{D}$ gives us the desired upper bound.

Next, let us prove the lower bound. For any $0<\epsilon<1$, consider the strategy $D^{M,\epsilon}$ from Definition \ref{DDefn}
that pays out dividend
at constant rate $(1-\epsilon)(\lambda\mathbb{E}[Y_{1}]-\rho)$
after the surplus hits above $M$. 
Then, 
\begin{align}
&\sup_{D\in\mathcal{D}}\mathbb{E}_{x}\left[\int_{0}^{T\wedge\tau}e^{-\delta t}dD_{t}\right]
\\
&\geq\mathbb{E}_{x}\left[e^{-\delta\tau_{M}}1_{\tau_{M}<\tau_{0}\wedge T}\right](1-\epsilon)(\lambda\mathbb{E}[Y_{1}]-\rho)
\mathbb{E}\left[\int_{0}^{(T-\tau_{M})\wedge\tau^{\epsilon}}e^{-\delta t}dt\middle|X_{0}^{\epsilon}=M, \tau_{M}<\tau_{0}\wedge T\right]
\nonumber
\\
&=\mathbb{E}_{x}\left[e^{-\delta\tau_{M}}1_{\tau_{M}<\tau_{0}\wedge T}\right](1-\epsilon)(\lambda\mathbb{E}[Y_{1}]-\rho)\nonumber\\
&\qquad\qquad\qquad\qquad\qquad\cdot\frac{1}{\delta}\left(1-\mathbb{E}\left[e^{-\delta((T-\tau_{M})\wedge\tau^{\epsilon})}\middle|X_{0}^{\epsilon}=M, \tau_{M}<\tau_{0}\wedge T\right]\right)
\nonumber,
\end{align}
which implies that
\begin{equation}
\liminf_{\lambda\rightarrow\infty}
\frac{1}{\lambda}\sup_{D\in\mathcal{D}}\mathbb{E}_{x}\left[\int_{0}^{T\wedge\tau}e^{-\delta t}dD_{t}\right]
\geq
(1-\epsilon)\frac{\mathbb{E}[Y_{1}]}{\delta}\left(1-\mathbb{E}[e^{-\delta(T\wedge\tau^{\epsilon})}|X_{0}^{\epsilon}=M]\right).
\end{equation}
Now, by first letting $M\rightarrow\infty$ and then $\epsilon\rightarrow 0$, 
and following the same arguments as in (i), we get the desired lower bound.

(iii) Assume $\delta=0$ and $T>0$.
The upper bound is similar as in (ii) by using the function $U_1(x,t):=x+(\lambda\mathbb{E}[Y_1]-\rho){(T-t)}$. 
We can show that
\begin{equation}
\sup_{D\in\mathcal{D}}\mathbb{E}_{x}\left[\int_{0}^{T\wedge\tau}dD_{t}\right]
\leq x+(\lambda\mathbb{E}[Y_{1}]-\rho)T.
\end{equation}

The lower bound can be obtained similarly as in (ii). We can show that
\begin{equation}
\liminf_{\lambda\rightarrow\infty}
\frac{1}{\lambda}\sup_{D\in\mathcal{D}}\mathbb{E}_{x}\left[\int_{0}^{T\wedge\tau}dD_{t}\right]
\geq
(1-\epsilon)\mathbb{E}[Y_{1}]\mathbb{E}[T\wedge\tau^{\epsilon}|X_{0}^{\epsilon}=M].
\end{equation}
Now, letting $M\rightarrow\infty$, we have $\tau^{\epsilon}\rightarrow\infty$ in probability
and by bounded convergence theorem, we have
\begin{equation}
\liminf_{\lambda\rightarrow\infty}
\frac{1}{\lambda}\sup_{D\in\mathcal{D}}\mathbb{E}_{x}\left[\int_{0}^{T\wedge\tau}e^{-\delta t}dD_{t}\right]
\geq
(1-\epsilon)\frac{\mathbb{E}[Y_{1}]}{\delta}T.
\end{equation}
Since it holds for any $\epsilon>0$, we get the desired lower bound.
\end{proof}

\begin{remark}\label{Alternative}
Indeed, we can also have a nearly optimal discrete dividend strategy
for the large $\lambda$ regime as an alternative
to the continuous dividend strategy $D^{\epsilon}$ defined in the proof
of Theorem \ref{LargeLambdaThm}. Let us consider a particular strategy $D^{\ast}$
which is a barrier strategy with barrier $x>0$, the same as the initial surplus.
Then, 
\begin{align}
\sup_{D\in\mathcal{D}}\mathbb{E}_{x}\left[\int_{0}^{\tau}e^{-\delta t}dD_{t}\right]
&\geq\mathbb{E}_{x}\left[\int_{0}^{\tau}e^{-\delta t}dD_{t}^{\ast}\right]
\\
&\geq\sum_{n=1}^{\infty}\mathbb{E}_{x}\left[e^{-\tau_{x}^{(n)}\delta}
1_{N[\tau_{x}^{(n-1)},\tau_{x}^{(n-1)}+\Delta t]\geq 1}1_{Y_{n}>\rho\Delta t}(Y_{n}-\rho\Delta t)\right]
\nonumber
\end{align}
where $\tau_{x}^{(n)}$ is the $n$-th time that the process jumps above
the threshold $x$ and $\tau_{x}^{(0)}:=0$ and $Y_{n}$ are i.i.d. distributed
as $Y_{1}$ as before. Essentially, we provide a lower bound by counting
the events that there is a jump occur before $\Delta t$ after the process starts
at $x$ and that jump size is greater than $\rho\Delta t$ which guarantees
that the process jumps above the threshold $x$ and pays the dividend.
Therefore, we have
\begin{align}
&\sup_{D\in\mathcal{D}}\mathbb{E}_{x}\left[\int_{0}^{\tau}e^{-\delta t}dD_{t}\right]
\\
&\geq\mathbb{E}\left[1_{Y_{1}>\rho\Delta t}(Y_{1}-\rho\Delta t)\right]
\sum_{n=1}^{\infty}\left(\mathbb{E}_{x}\left[e^{-\tau_{x}^{(1)}\delta}
1_{N[0,\Delta t]\geq 1}1_{\tau_{x}^{(1)}=\inf\{t>0:N_{t}=1\}}\right]\right)^{n}.
\nonumber
\end{align}
It is easy to compute that
\begin{align}
\mathbb{E}_{x}\left[e^{-\tau_{x}^{(1)}\delta}
1_{N[0,\Delta t]\geq 1}1_{\tau_{x}^{(1)}=\inf\{t>0:N_{t}=1\}}\right]
&=\int_{0}^{\Delta t}e^{-\delta s}\lambda e^{-\delta s}ds
\\
&=\frac{\lambda}{\delta+\lambda}\left[1-e^{-(\lambda+\delta)\Delta t}\right],
\nonumber
\end{align}
which implies that
\begin{equation}
\sup_{D\in\mathcal{D}}\mathbb{E}_{x}\left[\int_{0}^{\tau}e^{-\delta t}dD_{t}\right]
\geq\mathbb{E}\left[1_{Y_{1}>\rho\Delta t}(Y_{1}-\rho\Delta t)\right]
\frac{\frac{\lambda}{\delta+\lambda}\left[1-e^{-(\lambda+\delta)\Delta t}\right]}
{1-\frac{\lambda}{\delta+\lambda}\left[1-e^{-(\lambda+\delta)\Delta t}\right]}.
\end{equation}
Therefore, for any $\Delta t>0$, we have
\begin{equation}
\liminf_{\lambda\rightarrow\infty}\frac{1}{\lambda}
\sup_{D\in\mathcal{D}}\mathbb{E}_{x}\left[\int_{0}^{\tau}e^{-\delta t}dD_{t}\right]
\geq\frac{1}{\delta}\mathbb{E}\left[1_{Y_{1}>\rho\Delta t}(Y_{1}-\rho\Delta t)\right].
\end{equation}
By letting $\Delta t\rightarrow 0$, we get the desired lower bound.
\end{remark}


\subsection{Small Running Cost ($\rho$) Regime}
\begin{proof}[Proof of Theorem \ref{SmallRhoThm} and Proposition \ref{epsilonOptimalSmallRho}]
(i) By our assumption, $\rho=0$.
Let $v(x)=x+\frac{\lambda\mathbb{E}[Y_{1}]}{\delta}$. Then, we can compute that
\begin{align}
\max\left\{\lambda\int_{0}^{\infty}[v(x+y)-v(x)]p(y)dy-\delta v(x),1-v'(x)\right\}
=\max\left\{-\delta x,0\right\}=0.
\end{align}
Therefore $v(x)=x+\frac{\lambda\mathbb{E}[Y_{1}]}{\delta}$ is a classical solution of the above problem. Classical verification theorem such as \cite[Theorem 8.4.1]{Fleming-Soner-book-06} shows that $v(x)\ge V(x;\rho)$ for $\rho=0$. Thus, it is sufficient to show that there is a sequence of strategies $D^\epsilon\in\mathcal{D}$ such that 
\[
\limsup_{\epsilon\to0}\mathbb{E}_{x}\left[\int_0^\tau e^{-\delta t}dD^\epsilon_t\right]=v(x).
\]
Recall the strategy $D^\epsilon$ which pays dividend $x-\epsilon$ at the beginning and any surplus above $\epsilon$ thereafter. Then, the ruin never occurs and 
\[
\begin{split}
\mathbb{E}_{x}\left[\int_0^\tau e^{-\delta t}dD^\epsilon_t\right]&=x-\epsilon+\mathbb{E}_{x}\left[\sum_{n=1}^\infty e^{-\delta \tau^{(n)}}Y_n\right]=x-\epsilon+\mathbb{E}[Y_1]\sum_{n=1}^\infty \mathbb{E}_{x}\left[e^{-\delta \tau^{(n)}}\right]\\
&=x-\epsilon+\mathbb{E}[Y_1]\sum_{n=1}^\infty \mathbb{E}_{x}\left[e^{-n\delta \tau^{(1)}}\right]
=x-\epsilon+\frac{\lambda\mathbb{E}[Y_1]}{\delta},
\end{split}
\]
where for $n\ge1$, $\tau^{(n)}$ is the time of the $n$th jump of the process $J$. Thus, \eqref{SmallRhoT=infty} holds.

(ii) For $\delta>0$ and finite $T>0$, 
let $v(t,x)=x+\frac{\lambda\mathbb{E}[Y_{1}]}{\delta}(1-e^{-\delta(T-t)})$. 
Then $v(T,x)=x$ and
\begin{align*}
&\max\left\{\frac{\partial v}{\partial t}(t,x)
+\lambda\int_{0}^{\infty}[v(t,x+y)-v(t,x)]p(y)dy-\delta v(t,x),1-\frac{\partial}{\partial x}v(t,x)\right\}
\\
&=\max\left\{-\delta x,0\right\}=0.
\end{align*}
Therefore $v(t,x)=x+\frac{\lambda\mathbb{E}[Y_{1}]}{\delta}(1-e^{-\delta(T-t)})$ is a classical solution. Therefore, one can repeat the above argument for function $v(t,x)$ to obtain \eqref{SmallRhoTfinitedelta>0}.

(iii) When $\delta=0$, the function $v(t,x)=x+\lambda\mathbb{E}[Y_{1}](T-t)$ should be used to obtain \eqref{SmallRhoTfinitedelta=0}.
\end{proof}


\subsection{Small Discount Rate ($\delta$) Regime}
\begin{proof}[Proof of Theorem \ref{SmallDeltaThm} and Proposition \ref{epsilonOptimalSmallDelta}]
(i) Let us first prove the upper bound.

Let us recall that
\begin{equation}
\min\left\{\delta V(x;\delta)+\rho V'(x;\delta)-\lambda\int_{0}^{\infty}[V(x+y;\delta)-V(x;\delta)]p(y)dy,V'(x;\delta)-1\right\}=0,
\end{equation}
with $V(0)=0$. 

Let us define $U(x):=\delta V(x;\delta)$ so that $U$ satisfies
\begin{equation}
\min\left\{\delta U(x)+\rho U'(x)-\lambda\int_{0}^{\infty}[U(x+y)-U(x)]p(y)dy,U'(x)-\delta\right\}=0,
\end{equation}
with $U(0)=0$. 

We want to show that for $\delta>0$ sufficiently small,
\begin{equation}\label{U:1:defn}
U(x)
\leq U_{1}(x):=(\lambda\mathbb{E}[Y_{1}]-\rho)\left(1-e^{-\alpha(\delta)x}\right)+\delta x,
\end{equation}
where $\alpha(\eta)$ is the largest positive root of
\begin{equation}\label{alpha:eta:defn}
F(\alpha)=\alpha\rho-\lambda\int_{0}^{\infty}[e^{-\alpha y}-1]p(y)dy+\eta=0.
\end{equation}
If $\eta>0$, $\alpha(\eta)$ exists. 
For $\eta<0$ sufficiently small, $F(\alpha)=0$ has at least one positive root.
In addition, $\alpha(\delta)$ is a continuous function.

Let us show that $U\leq U_{1}$. Consider an arbitrary admissible dividend strategy $D_{t}$, 
increasing c\`{a}dl\`{a}g function with $D_{0}=0$. 
By change of variable formula for processes of bounded variation (see, e.g., \cite[Theorem~II.31]{Protter}), we have
\begin{align}
U_{1}(x)&=\mathbb{E}_{x}[e^{-\delta(t\wedge\tau)}U_{1}(X_{t\wedge\tau})]\\
&
+\mathbb{E}_{x}\Bigg[\int_{0}^{t\wedge\tau}e^{-\delta s}\bigg(\rho U_{1}'(X_{s})-\lambda\int_{0}^{\infty}[U_{1}(X_{s}+y)-U_{1}(X_{s})]p(y)dy
+\delta U_{1}(X_{s})\bigg)ds\Bigg]
\nonumber
\\
&\qquad
+\mathbb{E}_{x}\left[\int_{0}^{t\wedge\tau}e^{-\delta s}U'_{1}(X_{s})dD_{s}^{c}\right]
\nonumber
\\
&-\mathbb{E}_{x}\left[\sum_{s\leq t\wedge\tau}e^{-\delta s}\left[U_{1}(X_{s+})-U_{1}(X_{s})-U'_{1}(X_{s})\Delta D_{s}\right]\right],
\nonumber
\end{align}
where $D_{s}^{c}$ is the continuous part of $D_{s}$ (see, e.g., \cite[Pg.~70]{Protter}).
We recall the definition of $U_{1}(x)$ in \eqref{U:1:defn},
where $\alpha(\delta)$ is defined in \eqref{alpha:eta:defn}. Direct calculations shows that
\begin{align}
&\delta U_{1}+\rho U'_{1}-\lambda\int_{0}^{\infty}(U_{1}(x+y)-U_{1}(x))p(y)dy=\delta x^{2}\geq 0,
\\
&U'_{1}(x)\geq\delta.
\end{align}
Thus 
\begin{equation}
U_{1}(x)\geq\mathbb{E}_{x}\left[e^{-\delta(t\wedge\tau)}U_{1}(X_{t\wedge\tau})\right]
+\delta\mathbb{E}_{x}\left[\int_{0}^{t\wedge\tau}e^{-\delta s}dD_{s}\right].
\end{equation}
Notice that here we used $U'_{1}\geq \delta$, $\Delta D_{s}\geq 0$, $X_{s+}\leq X_{s}$ and
\begin{equation}
\sum_{s\leq t\wedge\tau}\left[U_{1}(X_{s+})-U_{1}(X_{s})-U'_{1}(X_{s})\Delta D_{s}\right]e^{-\delta s}\leq
-\delta\sum_{s\leq t\wedge\tau}e^{-\delta s}\Delta D_{s}.
\end{equation}
By sending $t\uparrow+\infty$, we obtain
\begin{equation}
U_{1}(x)\geq\lim_{t\rightarrow\infty}\mathbb{E}_{x}\left[e^{-\delta(t\wedge\tau)}U_{1}(X_{t\wedge\tau})\right]
+\delta\mathbb{E}_{x}\left[\int_{0}^{\infty}e^{-\delta t}dD_{t}\right].
\end{equation}
Note that $\lim_{t\rightarrow\infty}\mathbb{E}_{x}\left[e^{-\delta(t\wedge\tau)}U_{1}(X_{t\wedge\tau})\right]=0$.
Thus, by taking supremum over $D\in\mathcal{D}$, we get
\begin{equation}
U_{1}(x)\geq\delta\sup_{D\in\mathcal{D}}\mathbb{E}_{x}\left[\int_{0}^{\infty}e^{-\delta t}dD_{t}\right]
=\delta V(x;\delta)=U(x).
\end{equation}
From $U(x)\leq U_{1}(x)$, it follows that
\begin{equation}
\limsup_{\delta\rightarrow 0}\delta V(x;\delta)\leq(\lambda\mathbb{E}[Y_{1}]-\rho)(1-e^{-\alpha x}).
\end{equation}

(ii) Now let us consider the lower bound (the Proposition \ref{epsilonOptimalSmallDelta} will also follow). 
For any $M>0$ and $\epsilon>0$, let us recall that the definition of the dividend strategy $D^{M,\epsilon}$ in Definition \ref{DDefn}:
no dividend is paid out until the first time that the process
jumps above $M$ and then it pays dividend with continuous rate $(1-\epsilon)(\lambda\mathbb{E}[Y_1]-\rho)$, and also recall the definitions of $\tau_{M}$, $X_{t}^{\epsilon}$, $\tau^{\epsilon}$, $\tau^{M,\epsilon}$, and $\tau_{0}$.
Therefore, we have
\begin{align}
&\sup_{D\in\mathcal{D}}\mathbb{E}_{x}\left[\int_{0}^{\tau}e^{-\delta t}dD_{t}\right]
\\
&\geq\mathbb{E}_{x}\left[\int_{0}^{\tau^{M,\epsilon}}e^{-\delta t}dD_{t}^{M,\epsilon}\right]
\nonumber
\\
&\geq\mathbb{E}_{x}\left[e^{-\delta\tau_{M}}1_{\tau_{M}<\tau_{0}}\right]
\mathbb{E}\left[\int_{0}^{\tau^{\epsilon}}e^{-\delta t}(1-\epsilon)(\lambda\mathbb{E}[Y_{1}]-\rho)dt\middle|X_{0}^{\epsilon}=M\right]
\nonumber
\\
&=\mathbb{E}_{x}\left[e^{-\delta\tau_{M}}1_{\tau_{M}<\tau_{0}}\right]
(1-\epsilon)(\lambda\mathbb{E}[Y_{1}]-\rho)\frac{1}{\delta}\left(1-\mathbb{E}\left[e^{-\delta\tau^{\epsilon}}\middle|X_{0}^{\epsilon}=M\right]\right),
\nonumber
\end{align}
where the second inequality above uses the simple fact that $X_{\tau_{M}}\geq M$
and a lower bound is given by starting $X_{0}^{\epsilon}=M$ rather than $X_{0}^{\epsilon}=X_{\tau_{M}}$.
Notice that
\begin{equation}
\lim_{\delta\rightarrow 0}\mathbb{E}\left[e^{-\delta\tau^{\epsilon}}\middle|X_{0}^{\epsilon}=M\right]
=\mathbb{P}(\tau^{\epsilon}<\infty|X_{0}^{\epsilon}=M),
\end{equation}
and $u(x):=\mathbb{P}(\tau^{\epsilon}<\infty|X_{0}^{\epsilon}=x)$
satisfies the equation
\begin{equation}
-\rho u'(x)-(1-\epsilon)(\lambda\mathbb{E}[Y_{1}]-\rho)u'(x)
+\lambda\int_{0}^{\infty}[u(x+y)-u(x)]p(y)dy=0,
\end{equation}
with the boundary condition $u(0)=1$, which implies that
\begin{equation}
\mathbb{P}(\tau^{\epsilon}<\infty|X_{0}^{\epsilon}=M)
=e^{-\alpha^{\epsilon}M},
\end{equation}
where $\alpha^{\epsilon}$ is the unique positive value
that satisfies the equation:
\begin{equation}
\rho\alpha^{\epsilon}-(1-\epsilon)(\lambda\mathbb{E}[Y_{1}]-\rho)\alpha^{\epsilon}
+\lambda\int_{0}^{\infty}[e^{-\alpha^{\epsilon}y}-1]p(y)dy=0.
\end{equation}
Moreover, we have
\begin{equation}
\lim_{\delta\rightarrow 0}\mathbb{E}_{x}\left[e^{-\delta\tau_{M}}1_{\tau_{M}<\tau}\right]
=\mathbb{P}_{x}(\tau_{M}<\tau).
\end{equation}
Hence, we have
\begin{equation}
\liminf_{\delta\rightarrow 0}\delta
\sup_{D\in\mathcal{D}}\mathbb{E}_{x}\left[\int_{0}^{\tau}e^{-\delta t}dD_{t}\right]
\geq
\mathbb{P}_{x}(\tau_{M}<\tau)(1-\epsilon)(\lambda\mathbb{E}[Y_{1}]-\rho)
\left(1-e^{-\alpha^{\epsilon}M}\right).
\end{equation}
Since it holds for any $M>x$, we can let $M\rightarrow\infty$ and it follows that
\begin{equation}
\liminf_{\delta\rightarrow 0}\delta
\sup_{D\in\mathcal{D}}\mathbb{E}_{x}\left[\int_{0}^{\tau}e^{-\delta t}dD_{t}\right]
\geq
\mathbb{P}_{x}(\tau=\infty)(1-\epsilon)(\lambda\mathbb{E}[Y_{1}]-\rho).
\end{equation}
Finally, notice that $\mathbb{P}_{x}(\tau=\infty)=1-e^{-\alpha x}$
and it holds for any $\epsilon>0$
and we can let $\epsilon\rightarrow 0$.
That yields the desired lower bound.
\end{proof}


\subsection{Large Time Horizon ($T$) Regime}


\begin{proof}[Proof of Theorem \ref{LargeT}]
For any $D\in\mathcal{D}$, 
\begin{align}
\left|\mathbb{E}_{x}\left[\int_{0}^{T\wedge\tau}e^{-\delta t}dD_{t}\right]
-\mathbb{E}_{x}\left[\int_{0}^{\tau}e^{-\delta t}dD_{t}\right]\right|
&=\mathbb{E}_{x}\left[\int_{T\wedge\tau}^{\tau}e^{-\delta t}dD_{t}\right]
\\
&\leq\mathbb{E}_{x}\left[\int_{T}^{\infty}e^{-\delta t}dD_{t}\right].
\nonumber
\end{align}
The total expected discounted dividends for $\rho\geq 0$
is bounded above by the case when $\rho=0$. In the case when $\rho=0$, 
there is no ruin risk and any surplus should be paid out immediately
to the shareholders. Therefore, for any $D\in\mathcal{D}$,
an upper bound of $\mathbb{E}_{x}\left[\int_{T}^{\infty}e^{-\delta t}dD_{t}\right]$
is given as follows. If no dividend is paid before time $T$ then the expected value
of the surplus at time $T$ is $x+\lambda\mathbb{E}[Y_{1}]T$ when $\rho=0$
and that the surplus is paid out at time $T$, after which, any surplus
is paid out immediately so that for any $D\in\mathcal{D}$,
\begin{equation}
\mathbb{E}_{x}\left[\int_{T}^{\infty}e^{-\delta t}dD_{t}\right]
\leq(x+\lambda\mathbb{E}[Y_{1}]T)e^{-\delta T}
+\lambda\mathbb{E}[Y_{1}]\int_{T}^{\infty}e^{-\delta t}dt,
\end{equation}
which yields the desired result.
\end{proof}

Now let us turn to the proof of Theorem \ref{LargeT_delta=0} which would be directly implied by Lemma \ref{LemmaI} and Lemma \ref{LemmaII}. 
First, let us prove that the optimal strategy is not to pay any dividend until the maturity $T$
if the company is not ruined by then and all the surplus at the maturity
is given to the shareholders as the dividends.

\begin{lemma}\label{LemmaI}
\begin{equation}
\sup_{D\in\mathcal{D}}\mathbb{E}_{x}\left[\int_{0}^{T\wedge\tau}dD_{t}\right]
=\mathbb{E}_{x}\left[X^0_{T\wedge\tau_0}\right],
\end{equation}
where 
$X^{0}_t:=x-\rho t+J_t$ and $\tau_0$ is the ruin time of process $X^0$. 
\end{lemma}

\begin{proof}
Let $D$ be an arbitrary admissible dividend strategy such that $X_{T\wedge \tau}=0$, where $\tau$ is the ruin time of process
\[
dX_t=-\rho dt+dJ_t-dD_t,\;\text{\rm and }\; X_0=x.
\]
Next,  define the dividend strategy $\tilde D$ by $\tilde D_t=0$ for $t<T\wedge \tau_0$ and $\tilde D_{T\wedge \tau_0}=X^0_{T\wedge \tau_0}$. For $t<T$, we have $X_t=X^0_t-D_t$. Then, 
\[
\mathbb{E}_{x}\Bigg[\int_0^{T\wedge \tau}dD_t\Bigg]=\mathbb{E}_{x}\left[X^0_{T\wedge\tau}\right]-\mathbb{E}_{x}\left[X_{T\wedge\tau}\right]=\mathbb{E}_{x}\left[X^0_{T\wedge\tau}\right].
\]
On the other hand, since $\lambda\mathbb{E}[Y_1]-\rho>0$, $X^0$ is a submartingale and thus, $\mathbb{E}_{x}\left[X^0_{T\wedge\tau}\right]\le \mathbb{E}_{x}[X^0_{T\wedge\tau_0}]$. Here we used the fact that $\tau_0\ge \tau$. This implies that 
\[
\mathbb{E}_{x}\Bigg[\int_0^{T\wedge \tau}dD_t\Bigg]\le \mathbb{E}_{x}\left[X^0_{T\wedge\tau_0}\right]=\mathbb{E}_{x}\Bigg[\int_0^{T\wedge\tau_0}d\tilde D_t\Bigg].
\]
\end{proof}


The above lemma asserts that when there is no discounting, paying surplus as dividend 
at the terminal time is an optimal strategy. This allows us to focus only on the terminal time dividend payment strategy. When the dividend is paid at terminal time, the expected dividend is equal to $\mathbb{E}_{x}[X^{0}_{T\wedge\tau_{0}}]$, which we will estimate in the next step as $T\rightarrow\infty$.

\begin{lemma}\label{LemmaII}
\begin{equation}
\mathbb{E}_{x}\left[X^0_{T\wedge\tau_{0}}\right]
=(\lambda\mathbb{E}[Y_{1}]-\rho)\left[(1-e^{-\alpha x})T+\frac{e^{-\alpha x}x}{\rho-\lambda\int_{0}^{\infty}e^{-\alpha y}yp(y)dy}+o(1)\right],
\end{equation}
as $T\rightarrow\infty$.
\end{lemma}

Before we proceed to the proof of Lemma~\ref{LemmaII}, 
let us first state and prove a technical lemma.

\begin{lemma}\label{TauLemma}
Under the condition that $\rho<\lambda\mathbb{E}[Y_{1}]$, we have
\begin{equation}
\mathbb{E}_{x}[\tau_{0}\cdot 1_{\tau_{0}<\infty}]=\frac{x}{\rho-\lambda\int_{0}^{\infty}e^{-\alpha y}yp(y)dy}.
\end{equation}
\end{lemma}

\begin{proof}[Proof of Lemma \ref{TauLemma}]
When $\mathbb{P}_{x}(\tau_{0}=\infty)>0$, first, we compute
\begin{equation}
v(x):=\mathbb{E}_{x}\left[e^{-\theta\tau_{0}}1_{\tau_{0}<\infty}\right],
\end{equation}
which satisfies the equation:
\begin{equation}
-\rho v'(x)+\lambda\int_{0}^{\infty}[v(x)-v(x)]p(y)dy-\theta v(x)=0,
\end{equation}
with the boundary condition $V(0)=1$. 

It is easy to see that $v(x)=e^{-\alpha(\theta)x}$, where $\alpha(\theta)>0$
is the unique solution to the equation:
\begin{equation}
\rho\alpha(\theta)+\lambda\int_{0}^{\infty}[e^{-\alpha(\theta)y}-1]p(y)dy-\theta=0.
\end{equation}
Differentiating with respect to $\theta$, we get
\begin{equation}\label{thetaZero}
\rho\alpha'(\theta)-\lambda\alpha'(\theta)\int_{0}^{\infty}e^{-\alpha(\theta)y}yp(y)dy-1=0.
\end{equation}
On the other hand,
\begin{equation}
\mathbb{E}_{x}[\tau_{0} 1_{\tau_{0}<\infty}]=\alpha'(0)xe^{-\alpha(0)x}.
\end{equation}
By letting $\theta=0$ in \eqref{thetaZero}, we conclude that
\begin{equation}
\mathbb{E}_{x}[\tau_{0} 1_{\tau_{0}<\infty}]=\frac{e^{-\alpha x}x}{\rho-\lambda\int_{0}^{\infty}e^{-\alpha y}yp(y)dy},
\end{equation}
where $\alpha>0$ is the unique solution to the equation:
\begin{equation}
\rho\alpha+\lambda\int_{0}^{\infty}[e^{-\alpha y}-1]p(y)dy=0.
\end{equation}
Hence,
\begin{equation}
\mathbb{E}_{x}[\tau_{0}|\tau_{0}<\infty]=\frac{x}{\rho-\lambda\int_{0}^{\infty}e^{-\alpha y}yp(y)dy}.
\end{equation}
\end{proof}

\begin{remark}
Note that in Lemma~\ref{TauLemma}, $\mathbb{E}_{x}[\tau_{0}\cdot 1_{\tau_{0}<\infty}]<\infty$
requires that the condition
\begin{equation}\label{subI}
\rho-\lambda\int_{0}^{\infty}e^{-\alpha y}yp(y)dy>0
\end{equation}
is satisfied. This can be easily checked as follows.
Recall that $\alpha$ is the unique positive solution
to the equation
\begin{equation}\label{subII}
\rho\alpha+\lambda\int_{0}^{\infty}[e^{-\alpha y}-1]p(y)dy=0.
\end{equation}
Substituting \eqref{subII} into \eqref{subI}, the condition \eqref{subI}
is equivalent to:
\begin{equation}
\int_{0}^{\infty}[1-e^{-\alpha y}-e^{-\alpha y}\alpha y]p(y)dy>0.
\end{equation}
Let $F(x):=1-e^{-x}-e^{-x}x$ for $x\geq 0$. It is easy
to compute that $F(0)=0$ and $F'(x)=e^{-x}x>0$, which implies
that $F(x)>0$ for any $x>0$ and hence \eqref{subI} holds.
\end{remark}

\begin{remark}
In Lemma~\ref{TauLemma}, for the case when $Y_{i}$ are exponentially distributed, say $p(y)=\nu e^{-\nu y}$ for some $\nu>0$, 
then, we can compute that $\alpha=\frac{\lambda}{\rho}-\nu$ and
\begin{equation}
\mathbb{E}_{x}[\tau_{0}|\tau_{0}<\infty]=\frac{x}{\rho(1-\frac{\rho\nu}{\lambda})}.
\end{equation}
\end{remark}

Now we are ready to prove Lemma~\ref{LemmaII}.

\begin{proof}[Proof of Lemma \ref{LemmaII}]
By change of variable formula for processes of bounded variation (see, e.g., \cite[Theorem~II.31]{Protter}), 
we can compute that
\begin{equation}
\mathbb{E}_{x}\left[X^{0}_{T\wedge\tau_{0}}\right]
=x+(\lambda\mathbb{E}[Y_{1}]-\rho)\mathbb{E}_{x}[T\wedge\tau_{0}].
\end{equation}
Note that 
\begin{equation}
\mathbb{E}_{x}[T\wedge\tau_{0}]=\mathbb{E}_{x}[\tau_{0}\cdot 1_{\tau_{0}<T}]
+T\mathbb{P}_{x}(\tau_{0}>T).
\end{equation}
By Lemma \ref{TauLemma}, we have
\begin{equation}
\mathbb{E}_{x}[\tau_{0}\cdot 1_{\tau_{0}<\infty}]=\frac{x}{\rho-\lambda\int_{0}^{\infty}e^{-\alpha y}yp(y)dy}.
\end{equation}
Moreover, 
\begin{equation}
\lim_{T\rightarrow\infty}\mathbb{P}_{x}(\tau_{0}>T)=\mathbb{P}_{x}(\tau_{0}=\infty)
=1-e^{-\alpha x}.
\end{equation}
Hence, we proved that
\begin{equation}
\lim_{T\rightarrow\infty}\frac{1}{T}\mathbb{E}_{x}[X_{T\wedge\tau_{0}}]
=(\lambda\mathbb{E}[Y_{1}]-\rho)(1-e^{-\alpha x}).
\end{equation}

Finally, notice that
$\mathbb{E}_{x}[\tau_{0}\cdot 1_{\tau_{0}<\infty}]<\infty$. Therefore,
\begin{equation}
\mathbb{E}_{x}[\tau_{0}\cdot 1_{\tau_{0}<\infty}]=\int_{0}^{\infty}\mathbb{P}_{x}(t\leq\tau_{0}<\infty)dt<\infty,
\end{equation}
which implies that
\begin{equation}
\mathbb{P}_{x}(\tau_{0}\geq T)-\mathbb{P}_{x}(\tau_{0}=\infty)=\mathbb{P}_{x}(T\leq\tau_{0}<\infty)=o\left(T^{-1}\right),
\end{equation}
as $T\rightarrow\infty$. Hence, we conclude that
\begin{equation}
\mathbb{E}_{x}[X_{T\wedge\tau_{0}}]
=(\lambda\mathbb{E}[Y_{1}]-\rho)\left[(1-e^{-\alpha x})T+\frac{e^{-\alpha x}x}{\rho-\lambda\int_{0}^{\infty}e^{-\alpha y}yp(y)dy}+o(1)\right],
\end{equation}
as $T\rightarrow\infty$.
\end{proof}


\begin{proof}[Proof of Proposition \ref{epsilonOptimalLargeT}]
For any $M>0$ and $\epsilon>0$, let us recall the definitions of $\tau_{M}$, $X_{t}^{\epsilon}$, $\tau^{\epsilon}$, $\tau^{M,\epsilon}$, and $\tau_{0}$.
Following the similar arguments as in the proof of Theorem \ref{SmallDeltaThm}
and Proposition \ref{epsilonOptimalSmallDelta}, we have
\begin{align}
&\sup_{D\in\mathcal{D}}\mathbb{E}_{x}\left[\int_{0}^{\tau\wedge T}dD_{t}\right]
\geq\mathbb{E}_{x}\left[\int_{0}^{\tau^{M,\epsilon}\wedge T}dD_{t}^{M,\epsilon}\right]
\nonumber
\\
&\geq\mathbb{P}_{x}(\tau_{M}<\tau_{0}\wedge T)
(1-\epsilon)(\lambda\mathbb{E}[Y_{1}]-\rho)\mathbb{E}[\tau^{\epsilon}\wedge(T-\tau_{M})|X_{0}^{\epsilon}=M, \tau_{M}<\tau_{0}\wedge T]
\nonumber
\\
&=\mathbb{P}_{x}(\tau_{M}<\tau_{0}\wedge T)
(1-\epsilon)(\lambda\mathbb{E}[Y_{1}]-\rho)
\nonumber
\\
&\qquad\qquad
\cdot\bigg(\mathbb{E}[(T-\tau_{M})1_{\tau^{\epsilon}=\infty}|X_{0}^{\epsilon}=M, \tau_{M}<\tau_{0}\wedge T]
\nonumber
\\
&\qquad\qquad\qquad\qquad
+\mathbb{E}[\tau^{\epsilon}\wedge(T-\tau_{M})1_{\tau^{\epsilon}<\infty}|X_{0}^{\epsilon}=M, \tau_{M}<\tau_{0}\wedge T]\bigg).
\nonumber
\end{align}
Therefore, we have
\begin{align}
&\liminf_{T\rightarrow\infty}\frac{1}{T}\sup_{D\in\mathcal{D}}\mathbb{E}_{x}\left[\int_{0}^{\tau\wedge T}dD_{t}\right]
\\
&\geq
\mathbb{P}_{x}(\tau_{M}<\tau_{0})(1-\epsilon)(\lambda\mathbb{E}[Y_{1}]-\rho)\mathbb{P}(\tau^{\epsilon}=\infty|X_{0}^{\epsilon}=M).
\nonumber
\end{align}
Hence, we conclude that
\begin{align}
&\lim_{\epsilon\rightarrow 0}\lim_{M\rightarrow\infty}\liminf_{T\rightarrow\infty}\frac{1}{T}\sup_{D\in\mathcal{D}}\mathbb{E}_{x}\left[\int_{0}^{\tau\wedge T}dD_{t}\right]
\\
&\geq
\mathbb{P}_{x}(\tau_{0}=\infty)(\lambda\mathbb{E}[Y_{1}]-\rho)=(1-e^{-\alpha x})(\lambda\mathbb{E}[Y_{1}]-\rho).
\nonumber
\end{align}
\end{proof}



\subsection{Large Discount Rate ($\delta$) Regime}
\begin{proof}[Proof of Theorem \ref{LargeDeltaThm}]
When the initial surplus is paid out completely at time $0$, 
this strategy gives value $x$. Therefore, this gives us a lower bound.

Next, let us prove the upper bound. 
Notice that the optimal strategy is the barrier strategy with
the optimal barrier $b$ and for any $x$
\begin{equation}
V(x;\delta)\leq x-b+\frac{\lambda\mathbb{E}[Y_{1}]-\rho}{\delta}
\leq x+\frac{\lambda\mathbb{E}[Y_{1}]-\rho}{\delta}.
\end{equation}
This gives us the upper bound.
\end{proof}

\subsection{Small Time Horizon ($T$) Regime}
\begin{proof}[Proof of Theorem \ref{SmallTThm}]
Let us first prove the upper bound. It is clear that
\begin{equation}
\sup_{D\in\mathcal{D}}\mathbb{E}_{x}\left[\int_{0}^{\tau\wedge T}e^{-\delta t}dD_{t}\right]
\leq\sup_{D\in\mathcal{D}}\mathbb{E}_{x}\left[\int_{0}^{\tau\wedge T}dD_{t}\right]
=\mathbb{E}_{x}[X_{\tau\wedge T}]=\mathbb{E}_{x}[X_{T}],
\end{equation}
for sufficiently small $T>0$, 
since when there is no discount factor, it is never optimal to pay dividend
and if no dividend is paid out, then the ruin time $\tau\geq\frac{x}{\rho}>T$
for $T$ sufficiently small. We can easily compute that
\begin{equation}
\mathbb{E}_{x}[X_{T}]=x+(\lambda\mathbb{E}[Y_{1}]-\rho)T.
\end{equation}
This gives us the upper bound. Now let us turn to the proof of the lower bound. 
Let us consider a dividend strategy $D^{\epsilon}$ such that $x-\epsilon$
is paid out at time $0$ and then the remaining
surplus is $\epsilon$ and no dividend is paid out. 
For any $T$ that is sufficiently small so that $T<\frac{\epsilon}{\rho}$, then, 
ruin will not occur before time $T$, i.e., $\tau>T$. By considering
this particular strategy, we have
\begin{align}
\sup_{D\in\mathcal{D}}\mathbb{E}_{x}\left[\int_{0}^{\tau\wedge T}e^{-\delta t}dD_{t}\right]
&\geq\mathbb{E}_{x}\left[\int_{0}^{\tau\wedge T}e^{-\delta t}dD^{\epsilon}_{t}\right]
\\
&=x-\epsilon+e^{-\delta T}\left[\epsilon+(\lambda\mathbb{E}[Y_{1}]-\rho)T\right]
\nonumber
\\
&=x-\delta T\epsilon+\epsilon O(T^{2})+(\lambda\mathbb{E}[Y_{1}]-\rho)T
+O(T^{2}),
\nonumber
\end{align}
which holds for $T<\frac{\epsilon}{\rho}$. Take $\epsilon=2\rho T$ for example
will give us the desired lower bound.
\end{proof}

\section*{Acknowledgements}
We are grateful to the editor and two anonymous referees for helpful comments and suggestions.
Arash Fahim gratefully acknowledges support from the National Science Foundation via the award
NSF-DMS-1447067. Lingjiong Zhu is grateful to the support from the National Science Foundation via the award
NSF-DMS-1613164.

\bibliographystyle{plainnat}
\bibliography{dual_risk.bib}

\end{document}